\documentclass[11pt]{article}
\usepackage[utf8]{inputenc}
\usepackage{amsmath}
\usepackage{amsthm, amssymb}
\usepackage{amsfonts}
\usepackage{natbib}
\usepackage{bm}
\usepackage{thm-restate}

\usepackage{graphicx}
\usepackage{subcaption}
\usepackage{float}
\usepackage{authblk}

\usepackage{braket}

\usepackage{color}   
\usepackage{hyperref}
\hypersetup{
    linktoc=all,     
    linkcolor=blue,  
}

\usepackage{framed}
\usepackage{setspace}
\usepackage{listings}
\usepackage{mathtools}

\usepackage[margin=1in]{geometry}

\usepackage{todonotes}

\usepackage[ruled,vlined]{algorithm2e}

\def\MMGF{\mathrm{MMGF}}
\def\Tr{\mathrm{Tr}}

\def\calH{\mathcal{H}}

\def\calK{\mathcal{K}}

\def\supp{\mathsf{supp}}

\theoremstyle{plain}
\newtheorem{theo}{Theorem}[section]
 
\newtheorem{theorem}[theo]{Theorem}
\newtheorem{lemma}[theo]{Lemma}

\newtheorem{corollary}[theo]{Corollary}

\newtheorem{proposition}[theo]{Proposition}

\newtheorem{definition}[theo]{Definition}
\newtheorem{example}[theo]{Example}

\theoremstyle{remark}
\newtheorem{remark}[theo]{Remark}

\usepackage{dsfont}

\usepackage{hyperref}
\usepackage{cleveref}

\usepackage{tikz}
\usetikzlibrary{arrows.meta,positioning,calc,shapes.geometric}

\newcommand{\edit}[1]{{\color{black} #1}}
\usepackage{authblk}

\begin{document}

\title{When Does Quantum Differential Privacy Compose?}

\author[1]{
Daniel Alabi
\thanks{Email: \url{alabid@illinois.edu}}
}

\author[2]{
Theshani Nuradha
\thanks{Email: \url{nuradha@illinois.edu}}
}
\affil[1, 2] {
University of Illinois at Urbana-Champaign
}

\date{}

\maketitle

\begin{abstract}
Composition is a cornerstone of classical differential privacy, enabling strong end-to-end guarantees for complex algorithms through composition theorems (e.g., basic and advanced).
In the quantum setting, however, privacy is defined operationally against arbitrary measurements, and classical composition arguments based on scalar privacy-loss random variables no longer apply.
As a result, it has remained unclear when meaningful composition guarantees can be obtained for quantum differential privacy (QDP).

In this work, we clarify both the limitations and possibilities of composition in the quantum setting.
We first show that classical-style composition fails in full generality for POVM-based approximate QDP: even quantum channels that are individually perfectly private can completely lose privacy when combined through correlated joint implementations.

We then identify a setting in which clean composition guarantees can be restored.
For tensor-product channels acting on product neighboring inputs, we introduce a \emph{quantum moments accountant} based on an operator-valued notion of privacy loss and a matrix moment-generating function.
Although the resulting R\'enyi-type divergence does not satisfy a data-processing inequality, we prove that controlling its moments suffices to bound measured R\'enyi divergence, yielding operational privacy guarantees against arbitrary measurements.
This leads to advanced-composition-style bounds with the same leading-order behavior as in the classical theory.

Our results demonstrate that meaningful composition theorems for quantum differential privacy require carefully articulated structural assumptions on channels, inputs, and adversarial measurements, and provide a principled framework for understanding which classical ideas do and do not extend to the quantum setting.
\end{abstract}

\tableofcontents

\section{Introduction}
\label{sec:introduction}

Differential privacy (DP) provides a rigorous framework for limiting information leakage about sensitive inputs under randomized data analysis~\citep{DworkMNS06,DworkKMMN06}. \edit{With the growth of quantum technologies and the generation of quantum data, it is crucial to protect sensitive information related to the data generated and the data that is fed into these developing hybrid classical--quantum systems.}
While the theory of classical differential privacy is by now mature, extending its guarantees to quantum information processing raises both conceptual and technical challenges~\citep{QDP_computation17}.
In fully or hybrid (i.e., classical--quantum) quantum processing pipelines, we need a definition that is native to that setting. The outputs of the algorithms and the processing techniques have to obey the laws of quantum mechanics. 
Classical DP does not capture this because the released object (if any) from quantum algorithms should be a quantum state, and an adversary may apply arbitrary measurements, including collective measurements across multiple outputs, to that state.

\edit{With that, quantum DP is defined in a way to ensure privacy for quantum states such that any two quantum states that are classified as neighbors are indistinguishable under all measurements up to a certain limit demanded by privacy parameters ~\citep{QDP_computation17,hirche2022quantum} (see also~\Cref{def:qdp-povm})}
In the quantum setting, outputs of quantum systems are quantum states rather than classical samples, adversaries may perform arbitrary measurements, and correlations or entanglement across multiple outputs can fundamentally alter distinguishability properties~\citep{NC00,KSV02}.
These features complicate even basic questions about composition, a cornerstone of the classical DP theory.

In classical DP, composition guarantees are derived by tracking a scalar \emph{privacy-loss random variable} whose moment generating function (MGF) is additive under independent composition~\citep{DworkRoVa10}.
This structure underlies both basic composition and advanced composition results, including the moments accountant framework of Abadi et al.~\citep{AbadiChGoMcMiTaZh16}.
In contrast, for quantum channels privacy loss is not a scalar quantity prior to measurement, and naively taking a supremum over all measurements destroys additivity since measured quantum divergences are not always additive or at least sub-additive.
As a result, classical composition arguments do not directly extend to quantum differential privacy (QDP).

Much of the existing literature on composition in quantum differential privacy has taken an explicitly adversary-centric viewpoint~\citep{hirche2022quantum, NuradhaGW24}.
In this line of work, composition is analyzed primarily through the lens of what an adversary may do to the outputs of multiple private mechanisms, most notably by quantifying over increasingly powerful classes of measurements~\citep{NuradhaGW24,nuradha2025MeasuredHS}.
While this perspective is natural from an operational security standpoint, it largely abstracts away how quantum systems are actually composed in practice.
In quantum information processing, composition is implemented by specific channel constructions (such as tensor-product composition, factorized releases, or correlated joint implementations) whose structural properties can be as consequential for privacy as the adversary's measurement capabilities~\cite{AaronsonR19}.
By explicitly separating the model of channel composition from the model of adversarial measurement, our work shifts part of the focus from ``what measurements are allowed'' to ``how mechanisms are combined,'' and shows that many composition phenomena in quantum differential privacy are driven as much by the structure of the joint channel as by the power of the adversary.

This work develops a principled framework for understanding when and how composition guarantees can be recovered in the quantum setting.
Our analysis makes two key points.
First, classical-style composition fails in full generality for POVM-based approximate QDP when correlated joint channels or entangled neighboring inputs are allowed.
Second, under carefully articulated structural assumptions (most notably tensor-product channels acting on product neighboring inputs) one can recover clean and quantitatively sharp composition guarantees via an operator-level analogue of the classical moments accountant.

To achieve this, we introduce a \emph{quantum moments accountant} (QMA) based on the privacy-loss operator and a matrix moment-generating function.
While the resulting R\'enyi-type divergence does not satisfy a data-processing inequality in general, we show that controlling its moments suffices to bound \emph{measured R\'enyi divergence}, which by definition captures worst-case distinguishability over all measurements.
This yields a direct route from operator-level moment bounds to operational $(\varepsilon,\delta)$-QDP guarantees.

Throughout the paper, we emphasize the importance of distinguishing between different models of composition.
We propose a strict hierarchy (i.e., tensor-product, factorized $\subsetneq$ general joint composition) and show that failures of classical composition arise precisely when moving beyond tensor-product structure.
Our results therefore clarify not only what is possible in quantum composition, but also why additional assumptions are unavoidable.

\begin{figure}[t]
\centering
\begin{tikzpicture}[
  font=\small,
  arrow/.style={-Latex, line width=0.9pt},
  box/.style={draw, rounded corners=10pt, align=left, inner sep=8pt, text width=5.4cm},
  title/.style={font=\bfseries},
  tag/.style={draw, rounded corners=6pt, inner sep=2pt, font=\footnotesize},
  note/.style={font=\footnotesize, align=left, text width=5.4cm},
  x=1cm, y=1cm
]

\node[box] (tp) {
\textbf{Tensor-product composition}\\[2pt]
\textbf{Inputs:} $\rho_1,\ldots,\rho_m$ on $\calH_1,\ldots,\calH_m$\\
\textbf{Channel:} $A^{\otimes}=A_1\otimes\cdots\otimes A_m$\\
\textbf{Output:} $\bigotimes_{i=1}^m A_i(\rho_i)$\\[2pt]
\textbf{Intuition:} independent mechanisms on independent inputs.
};

\node[box, below=10mm of tp] (fac) {
\textbf{Factorized composition}\\[2pt]
\textbf{Input:} one state $\rho$ on $\calH$\\
\textbf{Channel:} $A^{\mathrm{fac}}(\rho)=\sum_x p(x) \bigotimes_{i=1}^m \sigma_i^{\rho,x}$\\
\textbf{Output:} product across outputs (classical correlations through $\{p(x)\}_x $ and no quantum correlations)\\[2pt]
\textbf{Intuition:} multiple independent ``views'' of the same input.
};

\node[box, below=10mm of fac] (joint) {
\textbf{General joint composition}\\[2pt]
\textbf{Input:} one state $\rho$ on $\calH$\\
\textbf{Channel:} any \edit{completely positive trace preserving} (CPTP) $A^{\mathrm{joint}}: \mathcal{D}(\calH)\to \mathcal{D}(\bigotimes_i \calK_i)$\\
\textbf{Constraint:} $\Tr_{K_{\setminus i}} A^{\mathrm{joint}}(\rho)=A_i(\rho)$ for all $i$\\
\textbf{Output:} may be correlated / entangled.
};

\draw[arrow] (tp.south) -- node[right, font=\footnotesize] {$\subsetneq$?} (fac.north);
\draw[arrow] (fac.south) -- node[right, font=\footnotesize] {$\subsetneq$?} (joint.north);

\coordinate (tpR) at ($(tp.east)+(1.2,0)$);
\node[tag, anchor=west] at ($(tpR)+(0,0.55)$) {$\calH_1$};
\node[tag, anchor=west] at ($(tpR)+(0,0.10)$) {$\calH_2$};
\node[tag, anchor=west] at ($(tpR)+(0,-0.35)$) {$\cdots$};
\node[tag, anchor=west] at ($(tpR)+(0,-0.80)$) {$\calH_m$};

\node[tag, anchor=west] at ($(tpR)+(1.4,0.55)$) {$\calK_1$};
\node[tag, anchor=west] at ($(tpR)+(1.4,0.10)$) {$\calK_2$};
\node[tag, anchor=west] at ($(tpR)+(1.4,-0.35)$) {$\cdots$};
\node[tag, anchor=west] at ($(tpR)+(1.4,-0.80)$) {$\calK_m$};

\draw[arrow] ($(tpR)+(0.65,0.55)$) -- ($(tpR)+(1.25,0.55)$);
\draw[arrow] ($(tpR)+(0.65,0.10)$) -- ($(tpR)+(1.25,0.10)$);
\draw[arrow] ($(tpR)+(0.65,-0.80)$) -- ($(tpR)+(1.25,-0.80)$);

\node[note, anchor=west] at ($(tpR)+(0,-1.45)$) {Each input subsystem is processed separately; no cross-talk.};

\coordinate (facR) at ($(fac.east)+(1.2,0)$);
\node[tag, anchor=west] (Hin) at ($(facR)+(0,0)$) {$\calH$};
\node[tag, anchor=west] (K1) at ($(facR)+(1.6,0.55)$) {$\calK_1$};
\node[tag, anchor=west] (K2) at ($(facR)+(1.6,0.10)$) {$\calK_2$};
\node[tag, anchor=west] (Kd) at ($(facR)+(1.6,-0.35)$) {$\cdots$};
\node[tag, anchor=west] (Km) at ($(facR)+(1.6,-0.80)$) {$\calK_m$};

\draw[arrow] ($(Hin.east)+(0.15,0.00)$) -- ($(K1.west)+(-0.15,0.00)$);
\draw[arrow] ($(Hin.east)+(0.15,0.00)$) -- ($(K2.west)+(-0.15,0.00)$);
\draw[arrow] ($(Hin.east)+(0.15,0.00)$) -- ($(Km.west)+(-0.15,0.00)$);

\node[note, anchor=west] at ($(facR)+(0,-1.55)$) {Same input; outputs are classically correlated product states across $\calK_1,\ldots,\calK_m$.};

\coordinate (jointR) at ($(joint.east)+(1.2,0)$);
\node[tag, anchor=west] (Hin2) at ($(jointR)+(0,0)$) {$\calH$};
\node[tag, anchor=west] (J1) at ($(jointR)+(1.6,0.35)$) {$\calK_1$};
\node[tag, anchor=west] (J2) at ($(jointR)+(1.6,-0.05)$) {$\calK_2$};
\node[tag, anchor=west] (Jd) at ($(jointR)+(1.6,-0.45)$) {$\cdots$};
\node[tag, anchor=west] (Jm) at ($(jointR)+(1.6,-0.85)$) {$\calK_m$};

\draw[arrow] ($(Hin2.east)+(0.15,0.00)$) -- ($(jointR)+(1.35,0.00)$);
\draw[line width=0.9pt] ($(jointR)+(1.95,0.10)$) to[out=0,in=0,looseness=1.3] ($(jointR)+(1.95,-0.60)$);
\node[font=\footnotesize, anchor=west] at ($(jointR)+(2.15,-0.25)$) {correlations allowed};

\node[note, anchor=west] at ($(jointR)+(0,-1.45)$) {Outputs can be correlated/entangled while matching prescribed marginals.};

\end{tikzpicture}
\caption{A hierarchy for multi-output composition models.}
\label{fig:composition-hierarchy-clean}
\end{figure}
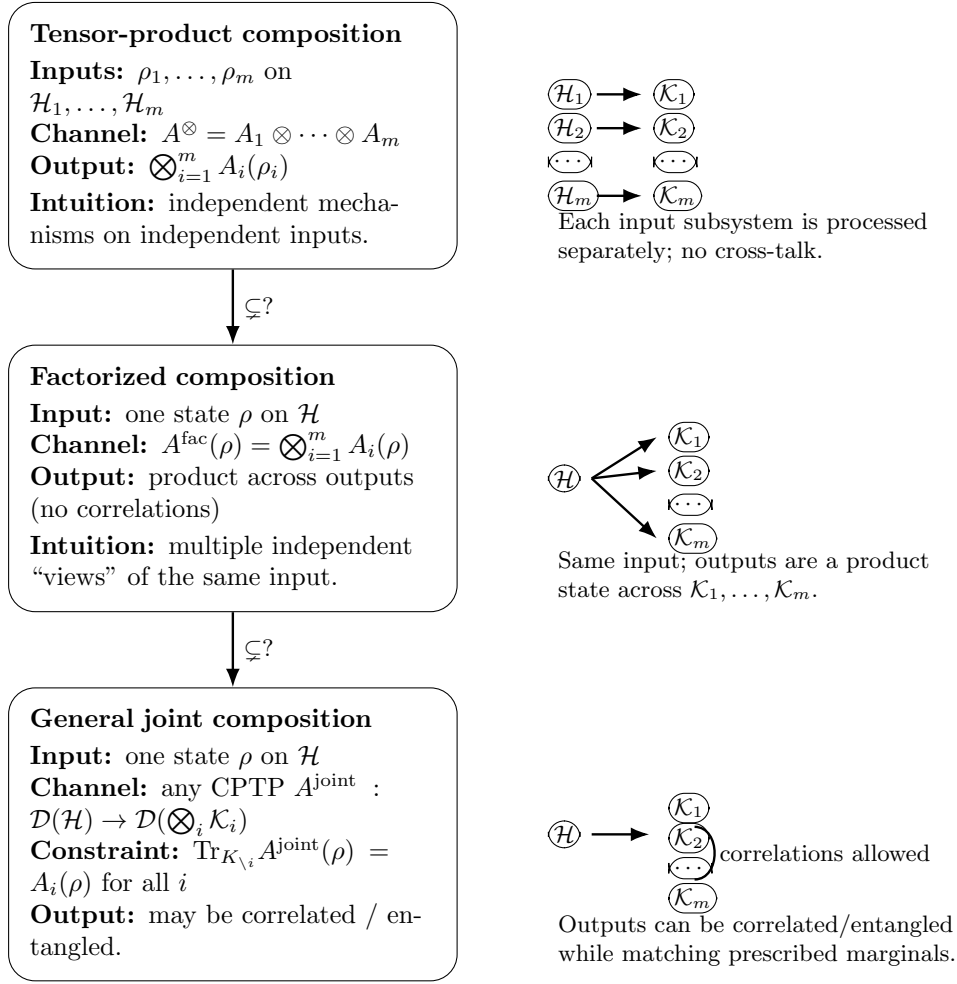

\subsection{Contributions}

This paper makes the following contributions.

\begin{itemize}
  \item \textbf{A hierarchy of quantum composition models.}
  We formalize and distinguish tensor-product, factorized, and general joint composition of quantum channels, and show that these form a strict hierarchy.
  This framework, illustrated in Figure~\ref{fig:composition-hierarchy-clean},
  isolates exactly where classical composition guarantees break down in the quantum setting. See Section~\ref{sec:models-of-composition}.

  \item \textbf{Impossibility of basic composition for general joint channels.}
  We prove, in Theorem~\ref{thm:no-go-basic-comp}, that even pure $(0,0)$-QDP mechanisms can fail to compose under correlated joint implementations.
  In particular, we show that, for $\varepsilon_1,\varepsilon_2,\delta_1,\delta_2\geq 0$ no general $(\varepsilon_1+\varepsilon_2,\delta_1+\delta_2)$-type composition theorem can hold for POVM-based approximate QDP without further restrictions.
  See Section~\ref{sec:no-basic-comp-qdp} for more details.

  \item \textbf{Safe basic composition under restricted adversaries.}
  For tensor-product channels on product neighboring inputs, we prove (in Theorem~\ref{thm:comp-oneway-locc-clean}) basic composition guarantees against \edit{one-way classical communication and local-operations-based (one-way LOCC)} adversaries, providing a clean quantum analogue of classical sequential composition under explicitly stated assumptions.
  See Section~\ref{sec:basic-comp-product} for more details.

  \item \textbf{A quantum moments accountant.}
  In Section~\ref{sec:qma}, we introduce a non-commutative analogue of the classical moments accountant based on the quantum privacy-loss operator and a matrix moment-generating function.
  We show, via Theorem~\ref{thm:qma-additivity-product}, that this accountant composes additively under tensor-product channels and product neighbors.

  \item \textbf{From operator moments to operational privacy.}
  We prove (Theorem~\ref{thm:qma-to-measured-rdp}) that bounds on the quantum moments accountant imply bounds on measured R\'enyi divergence, yielding operational $(\varepsilon,\delta)$-QDP guarantees against arbitrary POVMs.
  This provides an advanced-composition-style bound with the familiar $\sqrt{\sum_i \varepsilon_i^2 \log(1/\delta)}$ scaling over
  $k$ $(\varepsilon_i, 0)$-QDP mechanisms. See Proposition~\ref{prop:thm1-from-qma} and Corollary~\ref{cor:qma-composition}.

  \item \textbf{Advanced composition for QDP.} We prove an advanced composition result for \edit{tensor-product composition} (Theorem~\ref{thm:tensor-product-advanced}) of $k$ $(\varepsilon_i,\delta_i)$-QDP mechanisms under all possible adversaries for $\delta_i=0$, with the composed mechanism satisfying the privacy parameter $\sqrt{\sum_i \varepsilon_i^2 \log(1/\delta)}$ for $\delta \in (0,1)$. Moreover, we also prove an advanced composition result that even holds for $\delta_i \neq 0$, when the adversary is allowed local operations and joint classical post-processing.
  See Proposition~\ref{prop:all_local_advanced}.

  \item \textbf{Clarifying the quantum–classical gap.}
  Our results identify measurement incompatibility and the absence of a joint classical probability space as the fundamental obstacles to classical composition arguments, rather than a failure of tensor-product additivity at the channel level. See Section~\ref{sec:models-of-composition}.
\end{itemize}

\section{Preliminaries and Definitions}

We collect notation, definitions, and structural distinctions that will be used throughout the paper.
Because composition behavior in the quantum setting depends sensitively on how channels are combined, how neighboring inputs are defined, and what adversarial measurements are permitted, we make these modeling choices explicit at the outset.
In particular, we distinguish several notions of multi-output composition that coincide classically but diverge sharply in the quantum setting.

\subsection{Notation and Conventions}

All Hilbert spaces are finite-dimensional.
For a Hilbert space $\mathcal{H}$, we write $\mathcal{D}(\mathcal{H})$ for the set of density operators on $\mathcal{H}$, i.e.,
positive semidefinite operators with unit trace.

\begin{definition}[Density operators]
Let $\mathcal{H}$ be a finite-dimensional Hilbert space. We write
$\mathcal{D}(\mathcal{H})
:= \bigl\{ \rho \in \mathsf{L}(\mathcal{H}) \;\big|\; 
\rho \succeq 0,\ \Tr(\rho)=1 \bigr\}$,
where $\mathsf{L}(\mathcal{H})$ denotes the set of linear operators on $\mathcal{H}$.
Elements of $\mathcal{D}(\mathcal{H})$ are called \emph{density operators} or \emph{quantum states}.
\label{def:density-operators}
\end{definition}

Given a composite system $\mathcal{H}_1 \otimes \cdots \otimes \mathcal{H}_m$, we write
$\mathrm{Tr}_{\mathcal{H}\setminus i}$ for the partial trace over all subsystems except $\mathcal{H}_i$.
For operators $X,Y$ we write $X \preceq Y$ to denote that $Y-X$ is positive semidefinite.

Throughout, logarithms are natural unless otherwise stated.
When defining operator expressions involving inverses or logarithms, we adopt the convention that the expression is $+\infty$ whenever the required support conditions fail.

\subsection{Channels}

Quantum channels are completely positive, trace-preserving (CPTP) maps between spaces of density operators.

\begin{definition}[Tensor-product channels]
Let $\mathcal{H}_1,\mathcal{H}_2, \ldots, \mathcal{H}_m$ and $\mathcal{K}_1,\mathcal{K}_2, \ldots, \mathcal{K}_m$ be finite-dimensional Hilbert
spaces. Let
$
A_i:\mathcal{D}(\mathcal{H}_i)\to\mathcal{D}(\mathcal{K}_i), \ i\in\{1,2, \ldots,m\},
$
be quantum channels (i.e., completely positive, trace-preserving maps).
The \emph{tensor-product channel} associated with $(A_1, \ldots, A_m)$ is the channel
$
A_1 \otimes \cdots \otimes A_m:\mathcal{D}(\mathcal{H}_1\otimes\cdots\otimes\mathcal{H}_m)
\;\longrightarrow\;
\mathcal{D}(\mathcal{K}_1\otimes\cdots\otimes\mathcal{K}_m)$
defined by
$
(A_1 \otimes \cdots \otimes A_m)(\rho_1\otimes\cdots\otimes\rho_m)
\;:=\;
A_1(\rho_1)\otimes\cdots\otimes A_m(\rho_m),$
and extended linearly to all inputs $\rho\in\mathcal{D}(\mathcal{H}_1\otimes\cdots\otimes\mathcal{H}_m)$.
{Equivalently, $A_1\otimes\cdots \otimes A_m$ is a CPTP map whose Kraus operators are all tensor products
of Kraus operators of $A_1, \ldots, A_m$.}

\end{definition}

\begin{definition}[Factorized channels]
Let $\mathcal{H}$ be a finite-dimensional Hilbert space and let
$\mathcal{K}_1,\ldots,\mathcal{K}_m$ be finite-dimensional Hilbert spaces.
A quantum channel
$
A:\mathcal{D}(\mathcal{H})\to \mathcal{D}(\mathcal{K}_1\otimes\cdots\otimes\mathcal{K}_m)
$
is called \emph{factorized} if 
{
$
A(\rho) =\sum_{x} p(x) \ \bigotimes_{i=1}^m \sigma_i^{\rho,x},
 $
where $\sigma_i^{\rho,x}$ is a quantum state for all $i$ and $x$, and $\{p(x)\}_x$ is a probability distribution with $p(x) \geq 0$ and $\sum_x p(x)= 1$.
In this case, the $m$ outputs are not entangled and 
independent given the input state. }
\end{definition}

Note that a factorized channel can measure the input, copy the classical data, and then prepare
quantum states conditioned on the measurement outcome. This has been done in previous work~\citep{NuradhaGW24}.

\begin{definition}[General joint channels]
Let $\mathcal{H}$ be a finite-dimensional Hilbert space and let
$\mathcal{K}_1,\ldots,\mathcal{K}_m$ be finite-dimensional Hilbert spaces.
A \emph{general joint channel} is any quantum channel
$A:\mathcal{D}(\mathcal{H})\to \mathcal{D}(\mathcal{K}_1\otimes\cdots\otimes\mathcal{K}_m)$
that is completely positive and trace-preserving, with no further structural restrictions.
Equivalently, a general joint channel may produce arbitrary correlations or entanglement among the
output subsystems $\mathcal{K}_1,\ldots,\mathcal{K}_m$, and need not admit any factorization or
tensor-product decomposition.
\end{definition}

\subsection{Models of Composition}
\label{sec:models-of-composition}

Classically, releasing multiple outputs of differentially private mechanisms implicitly defines a joint distribution over all outputs.
In the quantum setting, however, there is no unique joint state consistent with a collection of marginals, and different joint implementations can exhibit drastically different privacy behavior.
We therefore distinguish three composition models:
\begin{enumerate}
  \item \textbf{Tensor-product composition}, where independent channels act on independent input subsystems.
  \item \textbf{Factorized composition}, where multiple channels act on the input but the joint output is constrained to be a product state (not entangled).
  \item \textbf{General joint composition}, where only the marginal behavior of each output is fixed, and the joint output may be arbitrarily correlated or entangled.
\end{enumerate}
While these notions might coincide in the classical setting, they form a hierarchy in the quantum setting.
Much of the subtlety of quantum composition arises from the gap between factorized and general joint composition.

\begin{definition}[Composition into tensor-product channels]
Let
$A_i:\mathcal{D}(\mathcal{H}_i)\to\mathcal{D}(\mathcal{K}_i),
\ i=1,\ldots,m,$
be quantum channels acting on (possibly distinct) input systems.
The \emph{composition into a tensor-product channel} is the joint channel
$
A^{\otimes}
\;:=\;
A_1\otimes\cdots\otimes A_m
\;:\;
\mathcal{D}(\mathcal{H}_1\otimes\cdots\otimes\mathcal{H}_m)
\to
\mathcal{D}(\mathcal{K}_1\otimes\cdots\otimes\mathcal{K}_m),
$
defined by
$
A^{\otimes}(\rho_1\otimes\cdots\otimes\rho_m)
=
A_1(\rho_1)\otimes\cdots\otimes A_m(\rho_m),
$
and extended linearly to all joint inputs.
This models the composition of \emph{independent mechanisms acting on independent inputs}.
\label{def:tensor-product-composition}
\end{definition}

{Note that tensor-product composition is the quantum analogue of classical parallel composition.}

\begin{definition}[Composition into factorized channels]
The \emph{composition into a factorized channel} is the joint channel
$A^{\mathrm{fac}}
\;:\;
\mathcal{D}(\mathcal{H})
\to
\mathcal{D}(\mathcal{K}_1\otimes\cdots\otimes\mathcal{K}_m)$
defined by 
$A^{\mathrm{fac}}(\rho)  =\sum_{x} p(x) \ \bigotimes_{i=1}^m \sigma_i^{\rho,x}$
In this composition, all outputs are conditionally independent given the input
state (not entangled), but each output may reveal different information about that input.
\label{def:factorized-composition}
\end{definition}

\begin{definition}[Composition into general joint channels]
Let
$
A_i:\mathcal{D}(\mathcal{H})\to\mathcal{D}(\mathcal{K}_i),
\ i=1,\ldots,m,
$
be quantum channels with a common input space.
A \emph{composition into a general joint channel} is any quantum channel
$
A^{\mathrm{joint}}
\;:\;
\mathcal{D}(\mathcal{H})
\to
\mathcal{D}(\mathcal{K}_1\otimes\cdots\otimes\mathcal{K}_m)
$
whose marginals coincide with the individual channels, i.e.,
$
\Tr_{\mathcal{K}_{\setminus i}}\!\big(A^{\mathrm{joint}}(\rho)\big)
=
A_i(\rho)
\ \text{for all }\rho\in\mathcal{D}(\mathcal{H})\text{ and all }i,
$
where $\mathcal{K}_{\setminus i}:=\bigotimes_{j\neq i}\mathcal{K}_j$.
Such a composition may introduce arbitrary classical or quantum correlations (including
entanglement) among the outputs and strictly generalizes both tensor-product and factorized
compositions.
\label{def:joint-composition}
\end{definition}

\begin{remark}[Tensor-product vs. factorized composition]
\label{rem:tensor-vs-factorized}
Tensor-product and factorized composition impose constraints along different
axes. Tensor-product composition restricts how a joint channel is \emph{implemented},
requiring independent local mechanisms, but allows correlated or entangled outputs
on entangled inputs. In contrast, factorized composition restricts the \emph{output
structure}, requiring a product state for every system (possibly with classical correlations), but does not require that
the channel decompose as a tensor product of local channels.

On product inputs $\rho=\rho_1\otimes\cdots\otimes\rho_m$, the output of a
tensor-product channel is product:
$
(A_1\otimes\cdots\otimes A_m)(\rho)
=
A_1(\rho_1)\otimes\cdots\otimes A_m(\rho_m),
$
but this coincidence does not extend to entangled inputs.
\end{remark}


\subsubsection{Examples}

\begin{example}[Example of tensor-product channel]
Let $\mathcal{H}_1=\mathcal{H}_2=\mathbb{C}^2$ and
$\mathcal{K}_1=\mathcal{K}_2=\mathbb{C}^2$.
Let $A_1$ and $A_2$ be single-qubit depolarizing channels,
$
A_i(\rho) = (1-p_i)\rho + p_i \frac{I}{2}, \ i\in\{1,2\}.
$
Then the tensor-product channel $A_1\otimes A_2$ acts on two-qubit inputs as
$
(A_1\otimes A_2)(\rho_1\otimes\rho_2)
= A_1(\rho_1)\otimes A_2(\rho_2),
$
and introduces no correlations between the two output qubits beyond those already present in the
input. If the input is a product state, the output is also a product state.
\label{ex:tensor-product}
\end{example}

\begin{example}[Example of factorized channel] 
Let $\mathcal H=\mathbb C^2$ and let $\mathcal K_i=\mathbb C^2$ for $i \in \{1, \ldots, m\}$.
Consider measure and prepare channel $A$ that measures the state $\rho$ with the POVM $\{M_x\}_x$ and then depending on the classical outcome, prepare state $ \otimes_{i=1}^m \sigma_i^x$.
Then the joint channel \\ 
$
\mathcal A(\rho)= \sum_x \mathrm{Tr}(M_x\rho) \ \sigma_1^x \otimes \cdots \otimes \sigma_m^x
$
is a valid factorized channel.

\label{ex:factorized}
\end{example}

\begin{example}[Example of general joint channels]
Let $\mathcal{H}=\mathbb{C}^2$ and $\mathcal{K}_1=\mathcal{K}_2=\mathbb{C}^2$.
Define a channel $A$ that appends an ancilla qubit initialized to $\ket{0}$ and applies a CNOT gate:
\[
A(\rho)=\mathrm{CNOT}\,(\rho\otimes\ket{0}\!\bra{0})\,\mathrm{CNOT}^\dagger.
\]
The resulting joint output state on $\mathcal{K}_1\otimes\mathcal{K}_2$ is generally entangled.
This channel is a valid joint channel but is neither tensor-product nor factorized, since the two
output subsystems are correlated even when the input is pure.
\label{ex:general-joint}
\end{example}

\begin{remark}[Distinguishing the three notions]
Tensor-product channels describe independent mechanisms acting on independent inputs.
Factorized channels describe multiple independent outputs derived from the same input.
General joint channels allow arbitrary correlations or entanglement across outputs.
These distinctions are crucial when reasoning about composition and privacy guarantees in the
quantum setting.
\end{remark}


\begin{proposition}[Containment relations between composition models]
\label{prop:composition-containment}
Let $m\ge 2$ and let $\mathcal{H},\mathcal{H}_1,\ldots,\mathcal{H}_m$ and
$\mathcal{K}_1,\ldots,\mathcal{K}_m$ be finite-dimensional Hilbert spaces.

\begin{enumerate}
\item\textbf{Tensor-product implies factorization on product inputs.}
Let $A^\otimes := A_1 \otimes \cdots \otimes A_m$ be a tensor-product
composition of channels $A_i : \mathcal D(H_i) \to \mathcal D(K_i)$.
Then for every product input state$
\rho = \rho_1 \otimes \cdots \otimes \rho_m,$
the output of $A^\otimes$ is factorized:
$
A^\otimes(\rho)
=
A_1(\rho_1)\otimes\cdots\otimes A_m(\rho_m).$
Equivalently, when restricted to product inputs, a tensor-product channel
coincides with a factorized channel (notice that in this case we set the probability distribution to a singleton).

\item \textbf{Factorized $\subseteq$ general joint.}
Every factorized composition is a special case of general joint composition.
Indeed, if
$
A^{\mathrm{fac}}(\rho)  =\sum_{x} p(x) \ \bigotimes_{i=1}^m \sigma_i^{\rho,x},
$
then $A^{\mathrm{fac}}$ is a joint channel whose marginals are exactly the $A_i(\rho)=\sum_x p(x) \sigma_i^{x,\rho}$.

\item \textbf{Strictness of inclusions.}
Both inclusions above are strict:
\begin{itemize}
\item There exist factorized channels that are not tensor-product channels (e.g., multiple outputs that are independent and possibly only classically correlated for a given input state).
\item There exist general joint channels that are not factorized (e.g., channels that produce
entangled outputs while preserving the same marginals).
\end{itemize}
\end{enumerate}
\end{proposition}

\begin{proof}
Follows from Definition~\ref{def:tensor-product-composition},
Definition~\ref{def:factorized-composition}, Definition~\ref{def:joint-composition} and Example~\ref{ex:tensor-product},Example~\ref{ex:factorized},Example~\ref{ex:general-joint}.
\end{proof}

\subsection{Operational Quantum Differential Privacy}

In the classical setting, differential privacy is defined in terms of the output distributions of randomized algorithms.
In the quantum setting, the output of a mechanism is a quantum state, and privacy must therefore be defined operationally in terms of the statistics produced by measurements.
This leads naturally to a definition that quantifies over all POVM effects applied to the channel output~\citep{QDP_computation17,hirche2022quantum}:

\begin{definition}[(Approximate) Quantum Differential Privacy]
\label{def:qdp-povm}
A quantum channel (CPTP map) $A$ is \emph{$(\varepsilon,\delta)$-QDP} if for all neighboring
input states $\rho\sim\sigma$ and for every measurement operator $M$
(i.e., satisfying $0\preceq M\preceq I$),
\begin{equation}
\Tr\!\big(MA(\rho)\big)\ \le\ e^{\varepsilon}\Tr\!\big(MA(\sigma)\big)\ +\ \delta.
\end{equation}
\end{definition}
Importantly, this operational definition captures adversaries with unrestricted measurement power.
While alternative notions of quantum privacy can be defined (e.g., by restricting measurements or comparing states directly via trace distance~\cite{NuradhaGW24}), the general quantum differential privacy definition most directly mirrors the classical adversarial model and is the strongest notion considered in this work.

\subsubsection{Neighboring Inputs}

We use the standard notion of neighboring inputs adapted to quantum states.
Two density operators $\rho,\sigma \in \mathcal{D}(\mathcal{H})$ are said to be \emph{neighbors}, denoted $\rho \sim \sigma$, if they differ in the data of a single individual or record.
The precise definition is application-dependent and left abstract in this work; all results hold for any fixed neighboring relation.
When considering multi-system inputs, we distinguish between:
\begin{itemize}
  \item \emph{Product neighbors}, where $\rho=\bigotimes_i \rho_i$ and $\sigma=\bigotimes_i \sigma_i$ with $\rho_i \sim \sigma_i$ for each subsystem; and
  \item \emph{General neighbors}, which may be entangled or classically correlated across subsystems.
\end{itemize}
This distinction plays a critical role in our composition results and impossibility theorems.

\subsection{Roadmap}

The distinctions introduced above are not merely definitional.
In Section~\ref{sec:no-basic-comp-qdp}, we show that basic composition fails for approximate QDP under general joint composition.
In contrast, Sections~\ref{sec:basic-comp-product} and~\ref{sec:qma} show that clean composition guarantees can be recovered for tensor-product channels acting on product neighboring inputs, culminating in an advanced-composition-style bound via the quantum moments accountant.
First, we discuss some additional related work.

\section{Related Work}

\paragraph{Classical differential privacy and composition.}
Composition theorems are central to classical differential privacy~\citep{KairouzOhVi17}.
Basic composition and advanced composition results were established early in the literature, culminating in the moments accountant framework of Abadi et al.~\citep{AbadiChGoMcMiTaZh16} for tracking privacy loss in iterative algorithms such as DP-SGD.
R\'enyi differential privacy (RDP) and its variants further systematized composition by working directly with R\'enyi divergences~\citep{mironov2017renyi}.

\paragraph{Quantum differential privacy.} Quantum differential privacy (QDP) was first introduced in \citep{QDP_computation17}, and several variants of QDP considering distinguishability under all possible measurements (adversaries) have been studied in~\citep{AaronsonR19,hirche2022quantum,angrisani2023unifying,NuradhaGW24, angrisani_localModel25,guan2024optimal,gallage2025theory}. Furthermore, a variant that generalizes QDP by encoding domain knowledge and measurement capabilities (practical possibilities of the adversary) of the systems, known as quantum pufferfish privacy (QPP), was introduced in~\citep{NuradhaGW24} and further studied in~\citep{nuradha2025MeasuredHS} with the inspiration from its classical variants~\citep{KM14, nuradha2022pufferfishJ}. This variant highlights, in some cases, how certain mechanisms are private when we consider the limitations of the measurements that can be performed, in contrast to those mechanisms being non-private when all measurements are allowed.
Prior works have explored relationships between QDP mechanisms, quantum hypothesis testing, and quantum R\'enyi divergences, as well as connections to classical DP under commuting states~\citep{hirche2022quantum,NuradhaGW24,Farhad_QP_HT,angrisani_localModel25,nuradha2024contraction,Christoph2024sample,dasgupta2025quantum}. 

In terms of composition of private mechanisms, for the variant of quantum local differential privacy (Definition~\ref{def:qdp-povm} for $\delta=0$ and neighbors declared as $\rho \sim \sigma$ for all pairs of distinct states), (basic) composition of several such private mechanisms when allowed measurements are having locally motivated structures has been studied in~\citep{guan2024optimal}, and in~\cite{hirche2022quantum,NuradhaGW24} for QDP and QPP, some basic composition results were provided. {In the setting where one applies private channels one after the other sequentially, sequential composition results for QDP have been studied using strong data-processing inequalities \cite{hirche2022quantum,  nuradha2025nonLinear}. 
However, more generally, composition guarantees in the quantum setting remain far less understood in a unified way, particularly in the presence of entanglement and correlated outputs.}

\paragraph{R\'enyi divergences in quantum information.}
Quantum R\'enyi divergences, including the Petz and sandwiched variants~\citep{MDSFT13,BSW15,RSB24}, play a central role in quantum information theory~\citep{DW18}.
The sandwiched R\'enyi divergence satisfies a data-processing inequality and has been used to derive operational guarantees in a variety of settings.
In contrast, the MMGF-induced divergence we consider is tailored to moment accounting and exact additivity, rather than monotonicity under channels.

\paragraph{Limits of composition in non-classical settings.}
Impossibility results related to composition have appeared in other non-classical or non-i.i.d. settings, where joint distributions or correlated releases invalidate union-bound-based arguments~\citep{KairouzOhVi17}.
Our no-go results for general joint quantum channels fit squarely into this theme and provide a concrete quantum-mechanical explanation rooted in measurement incompatibility.
\paragraph{Our contribution in context.}
Relative to prior work, this paper provides the first systematic treatment of when classical-style composition can and cannot be recovered for quantum differential privacy.
By explicitly separating structural assumptions on channels, neighboring inputs, and adversaries, we reconcile negative results with positive composition theorems and introduce a moments-accountant-style framework that is both quantum-native and operationally meaningful.

\section{Why Basic Composition for Approximate DP Fails for General Joint Channels}
\label{sec:no-basic-comp-qdp}

\paragraph{Basic composition (classical).}
In the classical setting, if $M_1$ is $(\varepsilon_1,\delta_1)$-DP and $M_2$ is
$(\varepsilon_2,\delta_2)$-DP, then the joint release $(M_1,M_2)$ is
$(\varepsilon_1+\varepsilon_2,\delta_1+\delta_2)$-DP. The proof relies on representing
privacy loss as a scalar random variable and applying a union bound to the ``$\delta$-bad'' events.

\subsection{Failure of Basic Composition for Quantum DP}

We now show that the classical basic composition theorem does not extend to
Definition~\ref{def:qdp-povm}.

\begin{theorem}[No-go for basic composition under POVM-QDP]
\label{thm:no-go-basic-comp}
There exist quantum channels $A_1,A_2$ and neighboring inputs $\rho\sim\sigma$ such that:
\begin{enumerate}
  \item each of $A_1$ and $A_2$ is $(0,0)$-QDP, yet
  \item there exists a composition of $A_1, A_2$ into a joint channel (Definition~\ref{def:joint-composition}) that is not $(\varepsilon,\delta)$-QDP for any $\delta<1$.
\end{enumerate}
Consequently, no general rule of the form
\[
(\varepsilon_1,\delta_1)\text{-QDP} + (\varepsilon_2,\delta_2)\text{-QDP}
\Longrightarrow (\varepsilon_1+\varepsilon_2,\delta_1+\delta_2)\text{-QDP}
\]
can hold for POVM-based quantum differential privacy.
\end{theorem}

\begin{proof}
Let the database be a single bit $b\in\{0,1\}$ encoded by orthogonal states
$\rho_0=\ket{0}\!\bra{0}$ and $\rho_1=\ket{1}\!\bra{1}$, which we declare neighboring.
Define a joint channel $A_{12}$ with two-qubit output systems $A$ and $B$ by
\[
A_{12}(\rho_0) = \ket{\Phi^+}\!\bra{\Phi^+},
\qquad
A_{12}(\rho_1) = \ket{\Phi^-}\!\bra{\Phi^-},
\]
where $\ket{\Phi^\pm} := \frac{1}{\sqrt{2}}(\ket{00}\pm\ket{11})$ are Bell states.

Define the individual mechanisms by taking marginals:
$
A_1 := \Tr_B \circ A_{12},
\
A_2 := \Tr_A \circ A_{12}.
$
Since
$
\Tr_B(\ket{\Phi^+}\!\bra{\Phi^+}) = \Tr_B(\ket{\Phi^-}\!\bra{\Phi^-}) = \tfrac{I}{2},
$
we have $A_1(\rho_0)=A_1(\rho_1)$, and similarly $A_2(\rho_0)=A_2(\rho_1)$.
Hence for every POVM effect $M$,
$
\Tr(MA_i(\rho_0))=\Tr(MA_i(\rho_1)), \ i\in\{1,2\},
$
so both $A_1$ and $A_2$ are $(0,0)$-QDP.
However, the joint outputs $A_{12}(\rho_0)$ and $A_{12}(\rho_1)$ are orthogonal.
Let $M=\ket{\Phi^+}\!\bra{\Phi^+}$. Then
$
\Tr(MA_{12}(\rho_0))=1, \ \Tr(MA_{12}(\rho_1))=0.
$
If $A_{12}$ were $(\varepsilon,\delta)$-QDP with $\delta<1$, Definition~\ref{def:qdp-povm}
would imply $1\le\delta$, a contradiction.
\end{proof}

\begin{remark}[Clarification about Theorem~\ref{thm:no-go-basic-comp}]
The joint channel refers to the correlated implementation (see Definition~\ref{def:joint-composition}) given by the
single joint channel $A_{12}$ outputting both subsystems, not to the independent tensor product
channel $A_1\otimes A_2$ (see Definition~\ref{def:tensor-product-composition}).
\end{remark}

\paragraph{Source of the composition failure}
The failure of composition in Theorem~\ref{thm:no-go-basic-comp} does \emph{not}
arise from entangled neighboring inputs. The neighboring states
$\rho_0 = \ket{0}\!\bra{0}$ and $\rho_1 = \ket{1}\!\bra{1}$ are single-qubit,
orthogonal, and unentangled.
Rather, the pathology arises from allowing a general joint channel whose
marginals are fixed but whose joint action introduces correlations between
outputs. Although each marginal channel is perfectly private, the joint
channel can amplify distinguishability through correlated outputs.

\subsection{Why This Phenomenon Is Quantum}

Classical approximate DP relies on two properties:
\begin{enumerate}
  \item \emph{Existence of a single joint probability space} supporting all outcomes and
        all privacy-loss events.
  \item \emph{Scalar privacy loss}: the likelihood ratio is a random variable, and $\delta$
        controls the probability of large deviations via a union bound.
\end{enumerate}

In contrast, Definition~\ref{def:qdp-povm} quantifies over all POVMs \emph{after} the channel.
Different POVMs are generally incompatible and cannot be realized jointly.
As a result, there is no single outcome space on which ``bad events'' can be union-bounded.
This lack of joint measurability is the same structural feature responsible for Bell/CHSH
violations in quantum mechanics~\citep{NC00,KSV02}.

Basic composition for approximate DP fundamentally relies on classical probabilistic structure.
Under POVM-based quantum differential privacy, this structure is absent: local indistinguishability
does not imply joint indistinguishability. Consequently, classical $(\varepsilon_1+\varepsilon_2,
\delta_1+\delta_2)$ composition has no general quantum analogue without additional assumptions.

The impossibility of basic $(\varepsilon_1+\varepsilon_2,\delta_1+\delta_2)$ composition
for approximate QDP stems from measurement incompatibility and the absence of a joint
classical probability space, rather than from a failure of tensor-product additivity
at the channel level.

\section{Basic Composition for Product Neighbors and Tensor-Product Channels}
\label{sec:basic-comp-product}

\subsection{Setup}

Let $\mathcal{H}_1,\mathcal{H}_2$ be finite-dimensional Hilbert spaces. For $i\in\{1,2\}$,
let $A_i:\mathcal{D}(\mathcal{H}_i)\to\mathcal{D}(\mathcal{K}_i)$ be quantum channels (CPTP maps),
and define the tensor-product channel composition (as in Definition~\ref{def:tensor-product-composition})
$
A_{12} := A_1\otimes A_2:\mathcal{D}(\mathcal{H}_1\otimes\mathcal{H}_2)\to
\mathcal{D}(\mathcal{K}_1\otimes\mathcal{K}_2).
$

We use the POVM-based approximate quantum differential privacy definition (Definition~\ref{def:qdp-povm}).

\paragraph{Product-neighbor model.}
We say that $(\rho,\sigma)$ are \emph{product neighbors} if $\rho=\rho_1\otimes\rho_2$,
$\sigma=\sigma_1\otimes\sigma_2$, and $\rho_i\sim\sigma_i$ for each $i\in\{1,2\}$.

\subsection{A Composition Theorem}

The clean ``$\delta$-adds'' basic composition proof in the classical setting relies on the
fact that a joint event can be decomposed into conditional events on each release (equivalently,
one can implement any test sequentially without changing the underlying probability space).
In the quantum setting, this exact argument goes through \emph{verbatim} provided the
adversary's test on the joint output is \emph{separable}. In particular, any local operations and classical communications (LOCC) test is
separable. For \emph{arbitrary} global POVMs, the corresponding statement is not generally
known to hold under the POVM-based definition, and the classical proof technique does not
directly apply.

\begin{definition}[One-way LOCC two-outcome tests]
A two-outcome test on $K_1\otimes K_2$ is one-way LOCC ($K_1\to K_2$) if it can be implemented as:
\begin{itemize}
\item measure $K_1$ with a POVM $\{E_t\}_t$,
\item on outcome $t$, measure $K_2$ with the POVM $\{M_{2,t}, I- M_{2,t}\}$ such that $0\preceq M_{2,t}\preceq I$ and accept iff it accepts.
\end{itemize}
Equivalently, its acceptance probability on product states satisfies
$\Tr(M\,\xi_1\otimes\xi_2)=\sum_t \Tr(E_t\xi_1)\Tr(M_{2,t}\xi_2)$.
\label{def:one-way-locc}
\end{definition}

\begin{theorem}[Tensor-product composition on product neighbors against one-way LOCC tests]
\label{thm:comp-oneway-locc-clean}
Let $A_i:\mathcal{D}(H_i)\to\mathcal{D}(K_i)$ be quantum channels for $i\in\{1,2\}$.
Assume that $A_i$ is $(\varepsilon_i,\delta_i)$-QDP:
for all neighboring
$\rho_i\sim\sigma_i$ and all measurement operators $0\preceq M_i\preceq I$ on $K_i$,
\[
\Tr(M_iA_i(\rho_i)) \le e^{\varepsilon_i}\Tr(M_iA_i(\sigma_i))+\delta_i.
\]
Fix any \emph{product} neighboring pair
$\rho=\rho_1\otimes\rho_2, \ \sigma=\sigma_1\otimes\sigma_2, \ \rho_i\sim\sigma_i.$
Let $A_{12}:=A_1\otimes A_2$. Then for every one-way LOCC ($K_1\to K_2$) two-outcome test on
$K_1\otimes K_2$ in the sense of Definition~\ref{def:one-way-locc} (equivalently, for every measurement operator $0\preceq M\preceq I$
admitting the factorization in \eqref{eq:oneway-factorization} below), we have
\begin{equation}
\label{eq:oneway-comp-bound}
\Tr\!\big(M\,A_{12}(\rho)\big)
\;\le\;
e^{\varepsilon_1+\varepsilon_2}\Tr\!\big(M\,A_{12}(\sigma)\big)
\;+\; e^{\varepsilon_2}\delta_1 \;+\; \delta_2.
\end{equation}
By swapping the roles of the two subsystems (i.e., using one-way LOCC tests $K_2\to K_1$),
we also obtain the symmetric bound
$
\Tr\!\big(M\,A_{12}(\rho)\big)
\;\le\;
e^{\varepsilon_1+\varepsilon_2}\Tr\!\big(M\,A_{12}(\sigma)\big)
\;+\; \delta_1 \;+\; e^{\varepsilon_1}\delta_2.$
Consequently, against one-way LOCC adversaries (in either direction) the tensor-product channel
$A_1\otimes A_2$ is $(\varepsilon_1+\varepsilon_2,\delta_{\mathrm{comp}})$-QDP on product neighbors
for
$
\delta_{\mathrm{comp}}:=\min\{\,e^{\varepsilon_2}\delta_1+\delta_2,\ \delta_1+e^{\varepsilon_1}\delta_2\,\}.
$
\end{theorem}

\begin{proof}
Let $\omega_i:=A_i(\rho_i)$ and $\eta_i:=A_i(\sigma_i)$ for $i\in\{1,2\}$. Then
$A_{12}(\rho)=\omega_1\otimes\omega_2$ and $A_{12}(\sigma)=\eta_1\otimes\eta_2$.

Fix a one-way LOCC ($K_1\to K_2$) two-outcome test with acceptance operator $0\preceq M\preceq I$.
By~\cref{def:one-way-locc}, there exist a POVM $\{E_t\}_t$ on $K_1$ (so $E_t\succeq 0$ and
$\sum_t E_t=I$) and measurement operator $0\preceq M_{2,t}\preceq I$ on $K_2$ such that for all product states
$\xi_1\otimes\xi_2$,
\begin{equation}
\label{eq:oneway-factorization}
\Tr\!\big(M(\xi_1\otimes\xi_2)\big) \;=\; \sum_t \Tr(E_t\xi_1)\,\Tr(M_{2,t}\xi_2).
\end{equation}
Applying \eqref{eq:oneway-factorization} to $\omega_1\otimes\omega_2$ gives
$
\Tr\!\big(M(\omega_1\otimes\omega_2)\big)
= \sum_t \Tr(E_t\omega_1)\,\Tr(M_{2,t}\omega_2).
$
For each $t$, since $0\preceq M_{2,t}\preceq I$ and $A_2$ is $(\varepsilon_2,\delta_2)$-QDP, we have $
\Tr(M_{2,t}\omega_2)\le e^{\varepsilon_2}\Tr(M_{2,t}\eta_2)+\delta_2.$
Multiplying by $\Tr(E_t\omega_1)\ge 0$ and summing over $t$ yields
\begin{align}
\Tr\!\big(M(\omega_1\otimes\omega_2)\big)
&\le \sum_t \Tr(E_t\omega_1)\Big(e^{\varepsilon_2}\Tr(M_{2,t}\eta_2)+\delta_2\Big) \nonumber\\
&= e^{\varepsilon_2}\sum_t \Tr(E_t\omega_1)\Tr(M_{2,t}\eta_2)\;+\;\delta_2\sum_t \Tr(E_t\omega_1).
\label{eq:stepA2}
\end{align}
Since $\sum_t E_t=I$ and $\Tr(\omega_1)=1$, we have $\sum_t \Tr(E_t\omega_1)=1$, so
\begin{equation}
\label{eq:stepA2b}
\Tr\!\big(M(\omega_1\otimes\omega_2)\big)
\le e^{\varepsilon_2}\sum_t \Tr(E_t\omega_1)\Tr(M_{2,t}\eta_2)\;+\;\delta_2.
\end{equation}
Define the scalars $s_t:=\Tr(M_{2,t}\eta_2)\in[0,1]$ and define the operator
$
N \;:=\;\sum_t s_t E_t.
$
Because $0\le s_t\le 1$ and $\sum_t E_t=I$, it follows that $0\le N\le I$. Moreover, \\
$
\sum_t \Tr(E_t\omega_1)\Tr(M_{2,t}\eta_2)
= \sum_t \Tr(E_t\omega_1)\,s_t
= \Tr(N\omega_1).
$
Substituting into \eqref{eq:stepA2b} gives
\begin{equation}
\label{eq:reduceN}
\Tr\!\big(M(\omega_1\otimes\omega_2)\big)
\le e^{\varepsilon_2}\Tr(N\omega_1)+\delta_2.
\end{equation}

Now apply $(\varepsilon_1,\delta_1)$-QDP for $A_1$ to the measurement operator $N$ (valid since $0\le N\le I$):
$
\Tr(N\omega_1)\le e^{\varepsilon_1}\Tr(N\eta_1)+\delta_1.
$
Plugging into \eqref{eq:reduceN} yields
\begin{equation}
\label{eq:afterA1}
\Tr\!\big(M(\omega_1\otimes\omega_2)\big)
\le e^{\varepsilon_2}\Big(e^{\varepsilon_1}\Tr(N\eta_1)+\delta_1\Big)+\delta_2
= e^{\varepsilon_1+\varepsilon_2}\Tr(N\eta_1) + e^{\varepsilon_2}\delta_1 + \delta_2.
\end{equation}
Finally, apply \eqref{eq:oneway-factorization} to $\eta_1\otimes\eta_2$:
$ \\
\Tr\!\big(M(\eta_1\otimes\eta_2)\big)
= \sum_t \Tr(E_t\eta_1)\Tr(M_{2,t}\eta_2)
= \sum_t \Tr(E_t\eta_1)\,s_t
= \Tr(N\eta_1).
$
Substituting $\Tr(N\eta_1)=\Tr(M(\eta_1\otimes\eta_2))$ into \eqref{eq:afterA1} gives
$
\Tr\!\big(M(\omega_1\otimes\omega_2)\big)
\le e^{\varepsilon_1+\varepsilon_2}\Tr\!\big(M(\eta_1\otimes\eta_2)\big) + e^{\varepsilon_2}\delta_1 + \delta_2,
$
which is exactly \eqref{eq:oneway-comp-bound}.
The symmetric claim follows by repeating the same argument with the roles of the subsystems swapped.
\end{proof}

\begin{remark}
Theorem~\ref{thm:comp-oneway-locc-clean} avoids any appeal to separable decompositions of $M$.
It uses only the \emph{defining} sequential factorization of one-way LOCC tests and the one-shot
POVM-based QDP inequalities for each channel. This is the direct quantum analogue of the classical
sequential/interactive proof of basic composition.
\end{remark}

\begin{corollary}[Small-$\varepsilon$ simplification]
\label{cor:small-eps-simplification}
In the setting of Theorem~\ref{thm:comp-oneway-locc-clean}, let
$\varepsilon_{\max}:=\max\{\varepsilon_1,\varepsilon_2\}$. Then
\[
\delta_{\mathrm{comp}}
=\min\{\,e^{\varepsilon_2}\delta_1+\delta_2,\ \delta_1+e^{\varepsilon_1}\delta_2\,\}
\le e^{\varepsilon_{\max}}(\delta_1+\delta_2).
\]
In particular, if $\varepsilon_{\max}\le 1$, then using $e^x\le 1+2x$ for all $x\in[0,1]$,
$
\delta_{\mathrm{comp}}\ \le\ (1+2\varepsilon_{\max})(\delta_1+\delta_2).
$
\end{corollary}

\begin{proof}
Since $e^{\varepsilon_2}\le e^{\varepsilon_{\max}}$ and $e^{\varepsilon_1}\le e^{\varepsilon_{\max}}$,
\[
\delta_{\mathrm{comp}}
=\min\{\,e^{\varepsilon_2}\delta_1+\delta_2,\ \delta_1+e^{\varepsilon_1}\delta_2\,\}
\le e^{\varepsilon_{\max}}\delta_1+\delta_2
\le e^{\varepsilon_{\max}}(\delta_1+\delta_2),
\]
and similarly using the other order. This proves the first inequality.

If $\varepsilon_{\max}\le 1$, the elementary bound $e^x\le 1+2x$ on $[0,1]$ gives
$e^{\varepsilon_{\max}}\le 1+2\varepsilon_{\max}$, hence
$\delta_{\mathrm{comp}}\le e^{\varepsilon_{\max}}(\delta_1+\delta_2)
\le (1+2\varepsilon_{\max})(\delta_1+\delta_2).$
\end{proof}

{
\begin{remark}[Relation to~previous works]
One can also arrive at Theorem~\ref{thm:comp-oneway-locc-clean} by using the general result for all joint measurements for the QDP setting in~\citep[Corollary~III.3]{hirche2022quantum} using syb-additivity properties of hockey-stick divergence.
Theorem~\ref{thm:comp-oneway-locc-clean} can also be obtained by
invoking results in the Quantum Pufferfish Framework (i.e., defining potential secrets, discriminative pairs, data distributions, and measurements). 
Specifically, by using
quasi-subadditivity of the Datta-Leditzky information spectrum divergence~\citep{DattaL15}
in Proposition 2 of~\cite{NuradhaGW24}, the same result in
Theorem~\ref{thm:comp-oneway-locc-clean} can be obtained. In this work, we provide an alternative proof that works for one-way LOCC measurements.
\end{remark}
}

In this section, we obtained composition guarantees when the adversary is allowed to perform one-way LOCC two-outcome tests as in~\Cref{def:one-way-locc}. This can be understood as the composed mechanism satisfying a flexible private variant of QDP (also a special case of quantum pufferfish privacy in~\cite{NuradhaGW24}).
In particular,~\Cref{thm:comp-oneway-locc-clean} shows that the composed mechanism satisfies quantum differential privacy with the neighbors defined as $\rho_1\otimes\rho_2 \sim \sigma_1\otimes\sigma_2, \ \rho_i\sim\sigma_i$ and $\mathcal{M}$ having measurement operators corresponding to two-outcome tests belonging to the category of one-way LOCC in~\Cref{def:one-way-locc}.

\section{Quantum Moments Accountant}
\label{sec:qma}

In the classical setting, the \emph{moments accountant} tracks the log-moment
generating function (MGF) of the privacy-loss random variable and exploits its
additivity under independent composition.  In the quantum setting, the main
obstruction is that privacy loss cannot be represented as a scalar random
variable prior to measurement, and taking a supremum over all POVMs destroys
additivity. Nevertheless, we show that a natural non-commutative analogue of the MGF (i.e., defined using the
privacy-loss operator and a Petz/exponential R\'enyi divergence induced by the MMGF) restores exact additivity.
This yields a clean and composable
quantum version of the classical moments accountant over pure DP channels.

The results we provide here are over tensor-product channels.

\subsection{Quantum Privacy-Loss Operator}

For neighboring states $\rho\sim\sigma$, define the \emph{quantum privacy-loss
operator}
\begin{equation}
    L(\rho,\sigma)
   \coloneqq \log\!\bigl(\sigma^{-1/2}\rho\,\sigma^{-1/2}\bigr),
\end{equation}
which is well defined on the support of $\sigma$.  
For a quantum channel $A$, we write \\
$
    L_{A}(\rho,\sigma)
    \coloneqq L(A(\rho),\,A(\sigma)).
$
This operator is a natural non-commutative analogue of the classical privacy
loss $\log \frac{P(o)}{Q(o)}$ for classical probability measures $P, Q$ and outcome/event $o$.

\paragraph{Support convention.}
Throughout, when defining
$L(\rho,\sigma)=\log\!\big(\sigma^{-1/2}\rho\,\sigma^{-1/2}\big),$
we implicitly restrict to pairs $(\rho,\sigma)$ satisfying
$\supp(\rho)\subseteq\supp(\sigma)$. If this condition fails, we define
the corresponding privacy-loss operator and MMGF to be $+\infty$.
Consequently, $\alpha_A(\lambda)$ may take the value $+\infty$ for some channels.

\subsection{Matrix Moment-Generating Function}
\label{sec:mmgf-and-exp-renyi}

For $\lambda>0$, quantum channel $A$, and states $\rho \sim\sigma$, we define the \emph{matrix moment-generating function} (MMGF) of
the privacy-loss operator by
\begin{equation}
\label{eq:mmgf-def}
    \MMGF_{A}(\lambda;\rho,\sigma)
    \coloneqq
    \Tr\!\left[
        A(\sigma)^{1/2}\,
        \exp\!\bigl(\lambda L_{A}(\rho,\sigma)\bigr)\,
        A(\sigma)^{1/2}
    \right].
\end{equation}
Since $e^{\lambda L_{A}(\rho, \sigma)} = (A(\sigma)^{-1/2}A(\rho)A(\sigma)^{-1/2})^{\lambda}$,
we can rewrite this as
\begin{equation}
\label{eq:mmgf-as-exp-renyi}
    \MMGF_{A}(\lambda;\rho,\sigma)
    =
    \Tr\!\left[
        A(\sigma)\,
        \bigl(
          A(\sigma)^{-1/2}A(\rho)A(\sigma)^{-1/2}
        \bigr)^{\lambda}
    \right].
\end{equation}

\subsection{Quantum Moments Accountant}
\label{sec:qma-def}

We now define the quantum analogue of the classical moments accountant.

\begin{definition}[Quantum Moments Accountant]
For a quantum channel $A$ and $\lambda>0$, define the quantum moments
accountant (QMA) by
$
    \alpha_A(\lambda)
    \coloneqq
    \sup_{\rho\sim\sigma}
    \log \mathrm{MMGF}_A(\lambda;\rho,\sigma),
$
where
$
    \mathrm{MMGF}_A(\lambda;\rho,\sigma)
    \coloneqq
    \Tr\!\big[
        A(\sigma)^{1/2}
        \exp\!\big(\lambda L_A(\rho,\sigma)\big)
        A(\sigma)^{1/2}
    \big], 
    L_A(\rho,\sigma)
    \coloneqq
    \log\!\big(
        A(\sigma)^{-1/2} A(\rho) A(\sigma)^{-1/2}
    \big).
$
\label{def:qma}
\end{definition}
As we will discuss, Definition~\ref{def:qma} induces an exponential R\'enyi-type divergence that composes
additively under tensor products:

\begin{definition}[Petz--R\'enyi divergence \cite{P85,P86}]
Let $\rho$ and $\sigma$ be density operators on a finite-dimensional Hilbert space
$\mathcal{H}$, and let $\alpha\in(0,1)\cup(1,\infty)$.
The \emph{Petz--R\'enyi divergence} of order $\alpha$ is defined as
$
D^{\mathrm{Petz}}_\alpha(\rho\|\sigma)
\;\coloneqq\;
\frac{1}{\alpha-1}
\log
\Tr\!\left[
\rho^\alpha\,\sigma^{1-\alpha}
\right],
$
with the convention that $D^{\mathrm{Petz}}_\alpha(\rho\|\sigma)=+\infty$ if
$\operatorname{supp}(\rho)\nsubseteq \operatorname{supp}(\sigma)$.
\end{definition}

\paragraph{Relation to R\'enyi-type divergences induced by the MMGF.}
The quantum moments accountant is defined via the log moment generating function
\[
\log \mathrm{MMGF}_A(\lambda;\rho,\sigma)
=
\log \Tr\!\left[
A(\sigma)\left(A(\sigma)^{-1/2}A(\rho)A(\sigma)^{-1/2}\right)^{\lambda}
\right].
\]
It is convenient to associate to this quantity an \emph{MMGF-induced} (exponential) R\'enyi-type divergence:

\begin{definition}[MMGF-induced divergence]
Let $\rho$ and $\sigma$ be density operators on a finite-dimensional Hilbert space
$\mathcal{H}$, and let $\alpha\in(0,1)\cup(1,\infty)$.
For $\alpha>1$, the \emph{MMGF-induced divergence} of order $\alpha$ is defined as 
\begin{equation}
\label{eq:mmgf-induced-divergence}
D^{\mathrm{MMGF}}_{\alpha}(\rho\|\sigma)
\;\coloneqq\;
\frac{1}{\alpha-1}\log \Tr\!\left[
\sigma\left(\sigma^{-1/2}\rho\sigma^{-1/2}\right)^{\alpha-1}
\right],
\end{equation}
with the convention $D^{\mathrm{MMGF}}_{\alpha}(\rho\|\sigma)=+\infty$ if $\supp(\rho)\nsubseteq\supp(\sigma)$.
\end{definition}

With this notation, for $\alpha=1+\lambda$ we have the identity \\  $
\frac{1}{\lambda}\log \mathrm{MMGF}_A(\lambda;\rho,\sigma)
=
D^{\mathrm{MMGF}}_{1+\lambda}\!\big(A(\rho)\|A(\sigma)\big).$

While $D^{\mathrm{MMGF}}_\alpha$ composes additively under tensor products for product inputs
(Section~\ref{sec:additivity}), it need not satisfy a data-processing inequality in general.
To obtain an operational privacy guarantee against \emph{arbitrary} POVMs, we instead show in
Theorem~\ref{thm:qma-to-measured-rdp} that bounds on the MMGF moments imply a bound on \emph{measured} R\'enyi divergence,
which directly captures worst-case distinguishability over all measurements:

\begin{definition}[Measured R\'enyi divergence]
Let $\alpha>1$. The \emph{measured R\'enyi divergence} of order $\alpha$ is
$
D^{\mathrm{meas}}_\alpha(\rho\|\sigma)
\;\coloneqq\;
\sup_{M\in\mathrm{POVM}}
D_\alpha\!\big(P_M(\rho)\,\|\,P_M(\sigma)\big),
$
where $P_M(\rho)$ denotes the classical outcome distribution induced by measuring $\rho$ with POVM
$M$, and $D_\alpha(\cdot\|\cdot)$ is the classical R\'enyi divergence of order $\alpha$.
\end{definition}

\subsection{Additivity Under Tensor-Product Composition}
\label{sec:additivity}

Let $A_1,\dots,A_k$ act on disjoint subsystems, and let $
A^{(k)} \coloneqq \bigotimes_{i=1}^k A_i .
$
For neighboring product states
$\rho = \bigotimes_{i=1}^k \rho_i$ and 
$\sigma = \bigotimes_{i=1}^k \sigma_i$, we have
$
L_{A^{(k)}}(\rho,\sigma)
    = \sum_{i=1}^k L_{A_i}(\rho_i,\sigma_i),
$
where each summand acts on its own tensor factor.

Because exponentials of tensor-factor sums factorize, 
and because $\Tr(X\otimes Y)=\Tr(X)\Tr(Y)$,
we obtain the following:

\begin{lemma}[Additivity of the Matrix MGF]
\label{lemma:mmgf-additivity}
For all $\lambda>0$, \\
$
\MMGF_{A^{(k)}}(\lambda;\rho,\sigma)
    =
    \prod_{i=1}^k
        \MMGF_{A_i}(\lambda;\rho_i,\sigma_i).
$
\end{lemma}
Taking logarithms and suprema yields:

\begin{theorem}[Additivity of the quantum moments accountant under product neighbors]
\label{thm:qma-additivity-product}
Fix $\lambda>0$. For a channel $\mathcal A$ define the \emph{product-neighbor accountant}
\[
\alpha_{\mathcal A}^{\mathrm{prod}}(\lambda)
\coloneqq\sup_{\substack{\rho=\otimes_{i=1}^k\rho_i,\ \sigma=\otimes_{i=1}^k\sigma_i\\ \rho_i\sim\sigma_i\ \forall i}}
\log \mathrm{MMGF}_{\mathcal A}(\lambda;\rho,\sigma).
\]
Let $\mathcal A^{(k)}\coloneqq\bigotimes_{i=1}^k \mathcal A_i$ be a tensor-product channel on disjoint subsystems.
Then for all $\lambda>0$,
$
\alpha_{\mathcal A^{(k)}}^{\mathrm{prod}}(\lambda)=\sum_{i=1}^k \alpha_{\mathcal A_i}(\lambda).
$
Moreover, for the \emph{unrestricted} accountant $\alpha_{\mathcal A^{(k)}}(\lambda)$ of Definition~\ref{def:qma}, {with neighbors having the same dimension as $\otimes_{i=1}^k \rho_i$},
we always have the lower bound \\
$
\alpha_{\mathcal A^{(k)}}(\lambda)\ \ge\ \alpha_{\mathcal A^{(k)}}^{\mathrm{prod}}(\lambda)
\ =\ \sum_{i=1}^k \alpha_{\mathcal A_i}(\lambda).
$
\end{theorem}

\begin{proof}
Let $\rho=\otimes_i \rho_i$ and $\sigma=\otimes_i \sigma_i$ be product neighbors.
Then $\mathcal A^{(k)}(\rho)=\otimes_i \mathcal A_i(\rho_i)$ and $\mathcal A^{(k)}(\sigma)=\otimes_i \mathcal A_i(\sigma_i)$.
The privacy-loss operators add across tensor factors, hence their exponentials factorize.
Using $\mathrm{Tr}(X\otimes Y)=\mathrm{Tr}(X)\mathrm{Tr}(Y)$ yields
$
\mathrm{MMGF}_{\mathcal A^{(k)}}(\lambda;\rho,\sigma)
=\prod_{i=1}^k \mathrm{MMGF}_{\mathcal A_i}(\lambda;\rho_i,\sigma_i),
$
so taking $\log$ gives additivity for each fixed product neighbor pair.
Taking the supremum over product neighbors yields
$\alpha_{\mathcal A^{(k)}}^{\mathrm{prod}}(\lambda)=\sum_i \alpha_{\mathcal A_i}(\lambda)$.
Finally, since the unrestricted supremum is over a larger set,
$\alpha_{\mathcal A^{(k)}}(\lambda)\ge \alpha_{\mathcal A^{(k)}}^{\mathrm{prod}}(\lambda)$.
\end{proof}

This mirrors the classical moments accountant exactly, but crucially
\emph{without requiring any measurement}.

\begin{proposition}[Advanced composition via the quantum moments accountant (measured R\'enyi route)]
\label{prop:thm1-from-qma}
Let $A_1,\ldots,A_k$ be quantum channels and let
$
A^{(k)} \;\coloneqq\; A_1\otimes\cdots\otimes A_k
$
denote their tensor-product composition. Consider product neighboring inputs
$\rho=\bigotimes_{i=1}^k \rho_i$ and $\sigma=\bigotimes_{i=1}^k \sigma_i$ with $\rho_i\sim\sigma_i$.
Fix $\alpha>1$ and write $\lambda\coloneqq\alpha-1$.

Assume that for each $i\in\{1,\ldots,k\}$ there exist parameters $\varepsilon_i\in(0,1]$,
$c_i\ge 0$, and $\alpha_i>1$ such that for all $\alpha\in(1,\alpha_i]$ and all $\rho_i\sim\sigma_i$,
\begin{equation}
\label{eq:qma-small-lambda-assumption}
\log \Tr\!\left[
A_i(\sigma_i)\left(A_i(\sigma_i)^{-1/2}A_i(\rho_i)A_i(\sigma_i)^{-1/2}\right)^{\alpha}
\right]
\;\le\;
\frac{1}{2}\varepsilon_i^2(\alpha-1)^2 \;+\; c_i(\alpha-1)^3.
\end{equation}
Let
\[
S \coloneqq \sum_{i=1}^k \varepsilon_i^2,
\qquad
C \coloneqq \sum_{i=1}^k c_i,
\qquad
\bar\alpha \coloneqq \min_{i\in[k]} \alpha_i,
\qquad
\bar\lambda \coloneqq \bar\alpha-1.
\]
Then for every $\delta\in(0,1)$ and every $\lambda\in(0,\bar\lambda]$ the channel $A^{(k)}$
satisfies $(\varepsilon(\lambda),\delta)$-QDP against \emph{arbitrary} POVMs, where
\begin{equation}
\label{eq:eps-lambda-updated}
\varepsilon(\lambda)
\;\coloneqq\;
\frac{S}{2}\lambda \;+\; C\lambda^2 \;+\; \frac{\log(1/\delta)}{\lambda}.
\end{equation}
In particular, letting
\[
\lambda^\star \coloneqq \sqrt{\frac{2\log(1/\delta)}{S}}
\quad\text{and}\quad
\hat\lambda \coloneqq \min\{\bar\lambda,\lambda^\star\},
\]
we obtain the explicit bound
\begin{equation}
\label{eq:eps-explicit-updated}
\varepsilon(\hat\lambda)
\;\le\;
\sqrt{2S\log(1/\delta)}
\;+\;
\frac{2C\log(1/\delta)}{S}
\;+\;
\frac{S}{2}\,(\bar\lambda-\lambda^\star)_+,
\end{equation}
where $(x)_+ \coloneqq \max\{x,0\}$.
\end{proposition}

\begin{proof}
Fix $\delta\in(0,1)$ and $\lambda\in(0,\bar\lambda]$, and set $\alpha\coloneqq1+\lambda$.

\paragraph{Step 1: Moment additivity under tensor products (product neighbors).}
For each $i$ define
\[
X_i \;\coloneqq\; A_i(\sigma_i)^{-1/2}A_i(\rho_i)A_i(\sigma_i)^{-1/2}\succeq 0.
\]
Since $A^{(k)}= \bigotimes_{i=1}^k A_i$ and the neighbors are products, we have
$
A^{(k)}(\rho)=\bigotimes_{i=1}^k A_i(\rho_i),
\\
A^{(k)}(\sigma)=\bigotimes_{i=1}^k A_i(\sigma_i),
$
and thus
$
X \;\coloneqq\; A^{(k)}(\sigma)^{-1/2}A^{(k)}(\rho)A^{(k)}(\sigma)^{-1/2}
=\bigotimes_{i=1}^k X_i.
$
Using $(\bigotimes_i X_i)^{\alpha}=\bigotimes_i X_i^{\alpha}$ and multiplicativity of the trace,
\begin{align}
\Tr\!\big(A^{(k)}(\sigma)\,X^{\alpha}\big)
&=
\Tr\!\left(\bigotimes_{i=1}^k A_i(\sigma_i)\;\bigotimes_{i=1}^k X_i^{\alpha}\right) =
\prod_{i=1}^k \Tr\!\big(A_i(\sigma_i)\,X_i^{\alpha}\big).
\end{align}
Taking logarithms and applying the bound \eqref{eq:qma-small-lambda-assumption} yields
\begin{equation}
\label{eq:composed-moment-bound}
\log \Tr\!\big(A^{(k)}(\sigma)\,X^{\alpha}\big)
\;\le\;
\sum_{i=1}^k \left(\frac{1}{2}\varepsilon_i^2\lambda^2 + c_i\lambda^3\right)
=
\frac{S}{2}\lambda^2 + C\lambda^3.
\end{equation}

\paragraph{Step 2: Convert the moment bound to measured R\'enyi DP.}
By Theorem~\ref{thm:qma-to-measured-rdp} applied to $A^{(k)}$ at order $\alpha=1+\lambda$,
for all product neighbors $\rho\sim\sigma$,
\[
D^{\mathrm{meas}}_{\alpha}\!\big(A^{(k)}(\rho)\,\|\,A^{(k)}(\sigma)\big)
\le
\frac{1}{\alpha-1}\log \Tr\!\big(A^{(k)}(\sigma)\,X^{\alpha}\big)
=
\frac{1}{\lambda}\log \Tr\!\big(A^{(k)}(\sigma)\,X^{\alpha}\big).
\]
Combining with \eqref{eq:composed-moment-bound} gives
\begin{equation}
\label{eq:measured-rdp-eps}
D^{\mathrm{meas}}_{1+\lambda}\!\big(A^{(k)}(\rho)\,\|\,A^{(k)}(\sigma)\big)
\;\le\;
\frac{1}{\lambda}\left(\frac{S}{2}\lambda^2 + C\lambda^3\right)
=
\frac{S}{2}\lambda + C\lambda^2.
\end{equation}
Equivalently, $A^{(k)}$ satisfies $(\alpha,\varepsilon_\alpha)$-measured R\'enyi DP with
$\varepsilon_\alpha=\frac{S}{2}\lambda + C\lambda^2$.

\paragraph{Step 3: Convert measured R\'enyi DP to $(\varepsilon,\delta)$-QDP.}
Fix an arbitrary POVM $M$. By definition of measured R\'enyi divergence,
\eqref{eq:measured-rdp-eps} implies that the induced classical distributions
$P_M(A^{(k)}(\rho))$ and $P_M(A^{(k)}(\sigma))$ satisfy
$
D_{\alpha}\!\big(P_M(A^{(k)}(\rho))\,\|\,P_M(A^{(k)}(\sigma))\big)
\le
\varepsilon_\alpha.
$
Applying the standard (classical) conversion from R\'enyi DP to approximate DP yields that for all
$\delta\in(0,1)$,
\[
\Pr[M(A^{(k)}(\rho))=1]
\le
\exp\!\left(\varepsilon_\alpha+\frac{\log(1/\delta)}{\alpha-1}\right)
\Pr[M(A^{(k)}(\sigma))=1]+\delta.
\]
Since $\alpha-1=\lambda$, we obtain $(\varepsilon(\lambda),\delta)$-QDP with
$
\varepsilon(\lambda)
=
\varepsilon_\alpha + \frac{\log(1/\delta)}{\lambda}
=
\left(\frac{S}{2}\lambda + C\lambda^2\right) + \frac{\log(1/\delta)}{\lambda},
$
which is exactly \eqref{eq:eps-lambda-updated}.

\paragraph{Step 4 (Optimization).}
For fixed $\delta\in(0,1)$, the conversion 
yields the bound
$
\varepsilon(\lambda)
\coloneqq \frac{\lambda}{2}\,S \;+\; C\lambda^2 \;+\; \frac{\log(1/\delta)}{\lambda},
\ \text{valid for } \lambda\in(0,\bar\lambda],
$
where $S\coloneqq\sum_{i=1}^k \varepsilon_i^2$.
Ignoring the constraint $\lambda\le \bar\lambda$ for a moment, the function
$
g(\lambda)\coloneqq\frac{\lambda}{2}\,S+\frac{\log(1/\delta)}{\lambda}
$
is minimized over $\lambda>0$ at
$
\lambda^\star \coloneqq \sqrt{\frac{2\log(1/\delta)}{S}}.
$
We therefore choose the feasible parameter
$
\hat\lambda \coloneqq \min\{\lambda^\star,\bar\lambda\}.
$
Substituting $\lambda=\hat\lambda$ gives the valid bound
\begin{equation}
\label{eq:eps-hatlambda-bound}
\varepsilon
\;\le\;
\frac{\hat\lambda}{2}\,S \;+\; C\hat\lambda^2 \;+\; \frac{\log(1/\delta)}{\hat\lambda}.
\end{equation}
In particular, if $\lambda^\star\le \bar\lambda$ then $\hat\lambda=\lambda^\star$ and
$
\varepsilon(\lambda^\star)
=
\sqrt{2S\log(1/\delta)}+\frac{2C\log(1/\delta)}{S}.
$
If instead $\lambda^\star>\bar\lambda$, then $\hat\lambda=\bar\lambda$ and we simply obtain the explicit feasible bound \\ 
$
\varepsilon(\bar\lambda)
=
\frac{\bar\lambda}{2}\,S + C\bar\lambda^2 + \frac{\log(1/\delta)}{\bar\lambda}.$
Combining these cases yields \eqref{eq:eps-hatlambda-bound} with
$\hat\lambda=\min\{\lambda^\star,\bar\lambda\}$.

\end{proof}

\begin{lemma}[Scalar Jensen bound for operator moments]
\label{lem:jensen-operator-moment}
Let $X\succeq 0$ be a positive semidefinite operator on a finite-dimensional Hilbert space and
let $\tau\in\mathcal{D}(\mathcal{H})$ be a density operator. Then for every $t\ge 1$,
$
\big(\Tr(\tau X)\big)^{t}\ \le\ \Tr\!\big(\tau X^{t}\big).
$
\end{lemma}

\begin{proof}
Let $X=\sum_j x_j \Pi_j$ be the spectral decomposition of $X$ with eigenvalues $x_j\ge 0$ and
orthogonal projectors $\Pi_j$. Define a probability distribution $r$ on eigen-indices by
$r_j \coloneqq \Tr(\tau \Pi_j)$, so that $r_j\ge 0$ and $\sum_j r_j = \Tr(\tau)=1$. Then
$
\Tr(\tau X)=\sum_j r_j x_j,
\
\Tr(\tau X^t)=\sum_j r_j x_j^{t}.
$
Since $f(x)=x^{t}$ is convex on $\mathbb{R}_+$ for $t\ge 1$, Jensen's inequality gives
$
\left(\sum_j r_j x_j\right)^t \le \sum_j r_j x_j^t,
$
which is exactly the desired inequality.
\end{proof}

\begin{theorem}[Quantum moments accountant $\Rightarrow$ measured R\'enyi DP]
\label{thm:qma-to-measured-rdp}
Fix $\alpha>1$ and let $A$ be a quantum channel. Define, for neighbors $\rho\sim\sigma$,
\[
X(\rho,\sigma)\;\coloneqq\;A(\sigma)^{-1/2}A(\rho)A(\sigma)^{-1/2},
\]
with the convention that if $\supp(A(\rho))\nsubseteq\supp(A(\sigma))$, then the expressions below
are $+\infty$. Suppose there exists a function $\alpha_A(\alpha)\ge 0$ such that for all neighbors
$\rho\sim\sigma$,
\begin{equation}
\label{eq:qma-bound-at-alpha}
\log \Tr\!\big(A(\sigma)\,X(\rho,\sigma)^{\alpha}\big)
\ \le\ (\alpha-1)\,\varepsilon_\alpha
\qquad\text{for some }\varepsilon_\alpha\ge 0.
\end{equation}
Then $A$ satisfies $(\alpha,\varepsilon_\alpha)$-\emph{measured R\'enyi DP}, i.e., for all
$\rho\sim\sigma$,
\[
D^{\mathrm{meas}}_\alpha\!\big(A(\rho)\,\|\,A(\sigma)\big)\ \le\ \varepsilon_\alpha.
\]
\end{theorem}

\begin{proof}
Fix neighbors $\rho\sim\sigma$ and an arbitrary POVM $M=\{M_z\}_z$ on the output space.
Let the induced classical distributions be $
p_z \coloneqq \Tr(M_z A(\rho)),\ q_z \coloneqq \Tr(M_z A(\sigma)).
$
If $\supp(A(\rho))\nsubseteq\supp(A(\sigma))$, then $\Tr(A(\sigma)X^\alpha)=+\infty$ and the claim is
trivial, so assume $\supp(A(\rho))\subseteq\supp(A(\sigma))$.

Define the positive operators
$
B_z \coloneqq A(\sigma)^{1/2} M_z A(\sigma)^{1/2}\succeq 0.
$
Then $\sum_z B_z = A(\sigma)$, and moreover
$
q_z = \Tr(B_z),\qquad p_z = \Tr(M_z A(\rho))=\Tr\!\big(B_z X(\rho,\sigma)\big).
$
For each $z$ with $q_z>0$, define the normalized state $\tau_z \coloneqq B_z/q_z\in\mathcal{D}$ so that
$
\frac{p_z}{q_z}=\Tr(\tau_z X).
$
Now consider the classical R\'enyi divergence of order $\alpha>1$:
\[
\exp\big((\alpha-1)D_\alpha(p\|q)\big)
=
\sum_z p_z^\alpha q_z^{1-\alpha}
=
\sum_z q_z\left(\frac{p_z}{q_z}\right)^\alpha
=
\sum_z q_z\big(\Tr(\tau_z X)\big)^\alpha.
\]
Apply Lemma~\ref{lem:jensen-operator-moment} with $t=\alpha$ to each term:
$
\big(\Tr(\tau_z X)\big)^\alpha \le \Tr(\tau_z X^\alpha).
$
Hence
\[
\sum_z q_z\big(\Tr(\tau_z X)\big)^\alpha
\le
\sum_z q_z\,\Tr(\tau_z X^\alpha)
=
\sum_z \Tr(B_z X^\alpha)
=
\Tr\!\left(\sum_z B_z\,X^\alpha\right)
=
\Tr\!\big(A(\sigma)\,X^\alpha\big).
\]
Therefore,
$
\exp\big((\alpha-1)D_\alpha(p\|q)\big)
\le
\Tr\!\big(A(\sigma)\,X^\alpha\big).
$
Taking logs and dividing by $\alpha-1$ gives \\
$
D_\alpha(p\|q)
\le
\frac{1}{\alpha-1}\log \Tr\!\big(A(\sigma)\,X^\alpha\big).$
By the assumed accountant bound \eqref{eq:qma-bound-at-alpha}, the RHS is at most $\varepsilon_\alpha$.
Since $M$ was arbitrary, taking the supremum over all POVMs yields \\
$
D^{\mathrm{meas}}_\alpha\!\big(A(\rho)\,\|\,A(\sigma)\big)
=
\sup_{M} D_\alpha\!\big(P_M(A(\rho))\,\|\,P_M(A(\sigma))\big)
\le
\varepsilon_\alpha,
$
which is the desired $(\alpha,\varepsilon_\alpha)$ measured R\'enyi DP guarantee.
\end{proof}

\begin{corollary}[Composition via the quantum moments accountant (measured R\'enyi route)]
\label{cor:qma-composition}
Let $A_1,\ldots,A_k$ be quantum channels and let
$
A^{(k)} \coloneqq A_1\otimes\cdots\otimes A_k
$
be their tensor-product composition. Consider product neighboring inputs
$\rho=\bigotimes_{i=1}^k \rho_i$ and $\sigma=\bigotimes_{i=1}^k \sigma_i$ with $\rho_i\sim\sigma_i$.
Fix $\alpha>1$ and set $\lambda\coloneqq\alpha-1$.

Assume each $A_i$ admits a quantum moments accountant at order $\alpha$, i.e.,
there exists $\alpha_{A_i}(\alpha)\ge 0$ such that for all $\rho_i\sim\sigma_i$,
\begin{equation}
\label{eq:alphaAi-def}
\log \Tr\!\left[
A_i(\sigma_i)\left(A_i(\sigma_i)^{-1/2}A_i(\rho_i)A_i(\sigma_i)^{-1/2}\right)^{\alpha}
\right]
\;\le\;
(\alpha-1)\,\alpha_{A_i}(\alpha).
\end{equation}
Then the composed channel $A^{(k)}$ satisfies $(\alpha,\varepsilon_\alpha)$-\emph{measured R\'enyi DP}
against arbitrary POVMs with
\begin{equation}
\label{eq:measured-rdp-composed}
\varepsilon_\alpha
=
\sum_{i=1}^k \alpha_{A_i}(\alpha).
\end{equation}
Consequently, for every $\delta\in(0,1)$, the composition satisfies $(\varepsilon',\delta)$-QDP with
\begin{equation}
\label{eq:qdp-from-composed}
\varepsilon'
=
\sum_{i=1}^k \alpha_{A_i}(\alpha)
\;+\;
\frac{\log(1/\delta)}{\alpha-1}.
\end{equation}
\end{corollary}

\begin{proof}
Fix $\alpha>1$ and write $\lambda\coloneqq\alpha-1$.

\paragraph{Step 1: Additivity of the $\alpha$-order moment bound under tensor products.}
Let $\rho=\bigotimes_{i=1}^k\rho_i$ and $\sigma=\bigotimes_{i=1}^k\sigma_i$ be product neighbors.
Define
\[
X_i \coloneqq A_i(\sigma_i)^{-1/2}A_i(\rho_i)A_i(\sigma_i)^{-1/2},
\qquad
X \coloneqq A^{(k)}(\sigma)^{-1/2}A^{(k)}(\rho)A^{(k)}(\sigma)^{-1/2}.
\]
Because $A^{(k)}= \bigotimes_{i=1}^k A_i$ and the inputs are products, we have 
$
A^{(k)}(\rho)=\bigotimes_{i=1}^k A_i(\rho_i),
\\
A^{(k)}(\sigma)=\bigotimes_{i=1}^k A_i(\sigma_i),
\
X=\bigotimes_{i=1}^k X_i.
$
Therefore, using $(\bigotimes_i X_i)^\alpha=\bigotimes_i X_i^\alpha$ and multiplicativity of the
trace,
\begin{align*}
\Tr\!\big(A^{(k)}(\sigma)\,X^{\alpha}\big)
&=
\Tr\!\left(\bigotimes_{i=1}^k A_i(\sigma_i)\;\bigotimes_{i=1}^k X_i^{\alpha}\right)
=
\prod_{i=1}^k \Tr\!\big(A_i(\sigma_i)\,X_i^{\alpha}\big),
\end{align*}
and hence
$
\log \Tr\!\big(A^{(k)}(\sigma)\,X^{\alpha}\big)
=
\sum_{i=1}^k \log \Tr\!\big(A_i(\sigma_i)\,X_i^{\alpha}\big).
$
Applying \eqref{eq:alphaAi-def} term-by-term yields \\
$
\log \Tr\!\big(A^{(k)}(\sigma)\,X^{\alpha}\big)
\le
(\alpha-1)\sum_{i=1}^k \alpha_{A_i}(\alpha).
$

\paragraph{Step 2: Convert the moment bound to measured R\'enyi DP.}
Applying Theorem~\ref{thm:qma-to-measured-rdp} at order $\alpha$ gives, for all product
neighbors $\rho\sim\sigma$, \\
$
D^{\mathrm{meas}}_\alpha\!\big(A^{(k)}(\rho)\,\|\,A^{(k)}(\sigma)\big)
\le
\frac{1}{\alpha-1}\log \Tr\!\big(A^{(k)}(\sigma)\,X^\alpha\big)
\le
\sum_{i=1}^k \alpha_{A_i}(\alpha),
$
which proves \eqref{eq:measured-rdp-composed}.

\paragraph{Step 3: Convert measured R\'enyi DP to $(\varepsilon',\delta)$-QDP.}
Fix any POVM $M$. By definition of measured R\'enyi divergence,
\[
D_\alpha\!\big(P_M(A^{(k)}(\rho))\,\|\,P_M(A^{(k)}(\sigma))\big)
\le
D^{\mathrm{meas}}_\alpha\!\big(A^{(k)}(\rho)\,\|\,A^{(k)}(\sigma)\big)
\le
\varepsilon_\alpha.
\]
Applying the standard classical conversion from $(\alpha,\varepsilon_\alpha)$-R\'enyi DP to
$(\varepsilon',\delta)$-DP yields
\[
\Pr[M(A^{(k)}(\rho))=1]
\le
e^{\varepsilon'}\Pr[M(A^{(k)}(\sigma))=1]+\delta
\quad\text{with}\quad
\varepsilon'=\varepsilon_\alpha+\frac{\log(1/\delta)}{\alpha-1}.
\]
Substituting \eqref{eq:measured-rdp-composed} gives \eqref{eq:qdp-from-composed}.
\end{proof}

\begin{remark}[Comparison to the classical moments accountant of Abadi et al.]
\label{rem:abadi-comparison}
Proposition~\ref{prop:thm1-from-qma} and Corollary~\ref{cor:qma-composition} are most
naturally compared to the classical moments accountant analysis of Abadi et al.\ for DP-SGD.

\paragraph{Classical moments accountant (Abadi et al.).}
In the classical setting one considers a sequence of (randomized) mechanisms and tracks the log
moment generating function of the \emph{privacy loss random variable}. Under adaptive composition,
these log-moments add, and one converts the resulting moment bound to an $(\varepsilon,\delta)$-DP
guarantee via Markov's inequality, followed by an optimization over the moment order.

\paragraph{Quantum analogue in this work.}
In the quantum setting, privacy is defined operationally against \emph{all} measurements (POVMs).
A direct quantum analogue of the classical privacy loss random variable is not available prior to
measurement. Instead, we work with an operator-level moment quantity:
for neighbors $\rho\sim\sigma$ we form the positive operator
$
X(\rho,\sigma)=A(\sigma)^{-1/2}A(\rho)A(\sigma)^{-1/2}
$
and track the moment functional $\Tr\!\big(A(\sigma)X(\rho,\sigma)^{\alpha}\big)$.
Theorem~\ref{thm:qma-to-measured-rdp} shows that controlling these moments implies an
\emph{operational} R\'enyi-type privacy guarantee, namely measured R\'enyi DP, which by definition
quantifies worst-case distinguishability over all POVMs.

\paragraph{Additivity vs.\ post-processing.}
As in the classical analysis, additivity under composition is obtained at the level of moments:
for tensor-product channels and product neighboring inputs, the moments factorize exactly by
tensor-product identities and multiplicativity of the trace, yielding the additive accountant in
Corollary~\ref{cor:qma-composition}.
Unlike the classical case, post-processing by measurements is handled \emph{inside} the privacy
notion via measured R\'enyi divergence: the supremum over POVMs is built into
$D^{\mathrm{meas}}_\alpha$, so no additional data-processing inequality is needed at this stage.

\paragraph{Resulting advanced-composition behavior.}
After converting measured R\'enyi DP to $(\varepsilon,\delta)$-QDP using the standard classical
R\'enyi-DP-to-approximate-DP conversion, one recovers an ``advanced composition'' tradeoff with
dominant term $\sqrt{\sum_i \varepsilon_i^2\log(1/\delta)}$ (cf.\ Abadi et al.).
The conceptual difference is that the quantum analysis separates (i) an operator-level accounting
step (moments add under tensor products) from (ii) an operational step (worst-case over POVMs),
whereas in the classical setting the privacy loss is classical from the outset.
\end{remark}

\section{Advanced Composition for Quantum Differential Privacy}
\label{sec:tensor-product-advanced}

In this section, we provide advanced composition results for QDP mechanisms under the tensor product composition (aka parallel composition) and neighboring notion based on product neighbors (i.e., $\otimes_{i=1}^k \rho_i \sim \otimes_{i=1}^k \sigma_i \iff \rho_i \sim \sigma_i \ \forall i \in \{1,\ldots, k\}$). 

For this setting, if one uses basic composition of $k$ QDP mechanisms, we would get a composed mechanism satisfying $(k\varepsilon, 0)$-QDP, if each of them satisfies $(\varepsilon, 0)$-QDP \cite{QDP_computation17, hirche2022quantum}. Next, we show that one can even obtain strong privacy guarantees that scale as $\sqrt{k} \varepsilon$ for sufficiently small $\varepsilon$, which may find use in high privacy regimes with $\varepsilon \leq 1$ and iterative protocols where the impact of $k$ is significant.

\begin{theorem}
Let $\varepsilon_i \in [0,1]$     for all $1\leq i \leq k$ and $\delta \in (0,1]$. Let the channel $A_i$ satisfy ($\varepsilon_i$, 0)-QDP. Then, the parallel composition (Definition~\ref{def:tensor-product-composition}) of
$A_1, \ldots, A_k$ (i.e., $A_1 \otimes \cdots \otimes A_k$) satisfies $(\varepsilon', \delta)$-QDP with
\begin{equation}
  \varepsilon' \coloneqq \frac{1}{2} \sum_{i=1}^k {\varepsilon_i^2} + \sqrt{2 \log\!\left( \frac{1}{\delta} \right) \sum_{i=1}^k {\varepsilon_i^2} }.
\end{equation}
Furthermore, for $\varepsilon_i=\varepsilon \leq 1$ for all $i$, we have that 
$
    \varepsilon' = \frac{ k \varepsilon^2}{2} + \sqrt{k} \varepsilon \sqrt{2 \log\!\left( \frac{1}{\delta} \right)}.
$
\label{thm:tensor-product-advanced}
\end{theorem}

\begin{proof}
Choose $D_\alpha$ satisfying data-processing and additivity for $\alpha>1$ (For example, Sandwiched R\'enyi divergence in Definition~\ref{def:sandwiched-renyi}). We also have $i \in \{1, \ldots, k\}$. 
$A_i$
    satisfying $(\varepsilon_i, 0)$-QDP \footnote{{In~\citep{NuradhaGW24}, the result is shown for the Quantum Pufferfish Privacy (QPP) framework,
    which generalizes the QDP definition.}}
    implies that we have from~\cite[Proposition~9, Proposition~13]{NuradhaGW24} 
    \begin{equation}
        \sup_{\rho \sim \sigma} D_\alpha( A_i (\rho) \Vert A_i(\sigma)) \leq \min \left\{ \frac{\varepsilon_i^2 \alpha}{2}, \varepsilon_i \right\}. \label{eq:i-renyi_bound}
    \end{equation}
  Then, consider the following:
\begin{align}
    \sup_{\rho_i \sim \sigma_i \ \forall i } D_{\alpha}\left( A_1(\rho_1)  \otimes \cdots \otimes A_k(\rho_k) \Vert A_1(\sigma_1)  \otimes \cdots \otimes A_k(\sigma_k) \right)  & =  \sup_{\rho_i \sim \sigma_i \ \forall i }  \sum_{i=1}^k D_{\alpha}\left( A_i(\rho_i)  \Vert A_i(\sigma_i)  \right) \\ 
   & \leq \sum_{i=1}^k  \min \left\{ \frac{\varepsilon_i^2 \alpha}{2}, \varepsilon_i \right\},
\end{align}
where the equality follows by the additivity of R\'enyi divergence and the inequality follows from~\eqref{eq:i-renyi_bound}.

Let $0 \preceq M_x \preceq I$ (note that $M_x$ is a measurement operator on $k$ sub-systems), define the following:
$
    p_x(\rho)  \coloneqq \Tr\!\left[M_x \left( A_1(\rho_1)  \otimes \cdots \otimes A_k(\rho_k)\right)\right], \
     q_x(\sigma) \coloneqq \Tr\!\left[M_x \left(A_1(\sigma_1)  \otimes \cdots \otimes A_k(\sigma_k) \right) \right].
$
Then, by the data-processing of R\'enyi divergence, we have that 
\begin{align}
    \sup_{\rho_i \sim \sigma_i \ \forall i } D_{\alpha}\left( A_1(\rho_1)  \otimes \cdots \otimes A_k(\rho_k) \Vert A_1(\sigma_1)  \otimes \cdots \otimes A_k(\sigma_k) \right)  \geq  \sup_{\rho_i \sim \sigma_i \ \forall i } \sup_{\{M_x\}_x} D_\alpha( p_x(\rho) \Vert q_x(\sigma)),
\end{align}
where the second supremization is over all POVMs over $k$-subsystems. 
So we have that for all $\{M_x\}_x$ and neighboring pairs, we have 
\begin{equation}
   D_\alpha( p_x(\rho) \Vert q_x(\sigma))   \leq \sum_{i=1}^k  \min \left\{ \frac{\varepsilon_i^2 \alpha}{2}, \varepsilon_i \right\} \leq \sum_{i=1}^k \frac{\varepsilon_i^2 }{2} \alpha = \zeta \alpha, \label{eq:upper_bound_with_Zeta}
\end{equation}
where $\zeta \coloneq \sum_{i=1}^k \varepsilon_i^2/2$. 
This is the point where we can utilize various classical procedures to obtain the desired composition result. 

By following the reasoning of the proof of~\cite[Proposition~3]{mironov2017renyi} by utilizing \cite[Proposition~10]{mironov2017renyi}, we have that 
\begin{equation}
    p_x(\rho) \leq e^{\varepsilon_\alpha} q_x(\sigma) + \delta
\end{equation}
for $\delta \in (0,1)$ and $\varepsilon_\alpha \coloneqq \zeta \alpha + \frac{\log(1/\delta)}{\alpha-1}$ (By choosing $D_\alpha$ as Sandwiched R\'enyi, we can also use Theorem~\ref{thm:srdp-to-qdp} to arrive at the above conclusion). This holds for all $0 \leq M_x \leq I$ and all neighboring pairs, so we have the guarantees of $(\varepsilon_\alpha, \delta)$-QDP. 

Note that the above analysis is valid for all $\alpha >1$. With that the privacy parameter can be optimized by obtaining 
$
    \varepsilon' = \min_{\alpha>1} \left(\zeta \alpha + \frac{\log(1/\delta)}{\alpha-1} \right). 
$
Let $f(\alpha'\equiv\alpha-1)= \zeta (\alpha'+1)  + b/\alpha'$ by denoting $b \equiv \log(1/\delta) $. Then, 
   $ f'(\alpha')=\zeta- b/(\alpha')^2.$
With that, the minimum is achieved at $    \alpha' = \sqrt{b/ \zeta}  $. 
So that, we have 
$
    \varepsilon' = \zeta + 2 \sqrt{ \zeta \log\!\left( \frac{1}{\delta} \right) },
$
and plugging that $\zeta= \sum_{i=1}^k \varepsilon_i^2 /2$, we conclude the proof.
\end{proof}

\subsection{Advanced Composition for Local Adversaries}

In the above sub-section, we prove an advanced composition result for $(\varepsilon, \delta)$-QDP, where $\delta=0$ when the adversary is allowed to choose any joint measurement on the composed system. It is not exactly clear how to generalize this result for $\delta \neq 0$ in general for all possible measurements applied by an adversary. Next, we look into obtaining stronger composition results under restrictions on the measurements (e.g.; locality of measurements) one could do on the composed quantum channel comprising $k$ separate quantum sub-systems composed in the tensor product fashion.

Let us consider the following measurement set that is related to local operations and classical post-processing:
\begin{equation}
    \mathcal{M}_{\operatorname{LO}^*} \coloneqq \left\{ \sum_{z_1,\ldots,z_k} T(z_1,\dots,z_k) \  M_1^{z_1}\otimes\cdots\otimes M_k^{z_k} :  \ \{M_i^{z_i}\}_{z_i\in\mathcal Z_i} \textnormal{ is a POVM } \forall i,  \ T(\cdot) \in [0,1] \right\}. 
\end{equation}
Note that the function $T(\cdot)$ can be understood as the classical processing component that takes each local measurement outcome as input.

\begin{example}[Local measurements and classical processing]
    $\mathcal{M}_{\operatorname{LO}^*} $ consists of all measurements one could do locally, independently on each sub-system, and then do classical post-processing. Note that it is not required for 
    $\sum_{z_1,\ldots,z_k} T(z_1,\dots,z_k) =1$, even though it is an special case. 
    
    As an example, consider $k=2$, and each sub-system is a qubit system, where each party performs the computational basis POVM (i.e.; $M^{0}_i =|0\rangle\!\langle 0|$ and $M^{1}_i =|1\rangle\!\langle 1|$ for all $i \in \{1,2\}$). Then, in the classical processing, it will accept iff both outcomes are not the same (i.e; $T(z_1, z_2)=1$ if $z_1 \neq z_2$ and $T(z_1, z_2)=0$ if $z_1 = z_2$). 
\end{example}

\begin{proposition}
    Let $A_i$ satisfies $(\varepsilon_i, 0)$-QDP. Then for product neighbors $ \otimes_{i=1}^k \rho_i \sim \otimes_{i=1}^k \sigma_i$ and $M\in \mathcal{M}_{\operatorname{LO}^*}$, we have that 
    \begin{equation}
        \Tr\!\left[ M \left( \bigotimes_{i=1}^k A_i(\rho_i) \right) \right] \leq e^{\bar{\varepsilon}} \Tr\!\left[ M \left( \bigotimes_{i=1}^k A_i(\sigma_i) \right) \right] + {\delta},
    \end{equation}
where for all $\delta > 0$ and
\begin{align}
    \bar{\varepsilon} & \coloneqq  \sum_{i=1}^k \varepsilon_i \left( \frac{e^{\varepsilon_i}-1}{e^{\varepsilon_i} +1} \right) + \sqrt{2 \log\!\left( \frac{1}{\delta} \right) \sum_{i=1}^k {\varepsilon_i^2} }. 
\end{align}
\end{proposition}
\begin{proof}

Fix $M\in\mathcal M_{\operatorname{LO}^*}$. By definition, there exist POVMs
$\{M_i^{z_i}\}_{z_i\in\mathcal Z_i}$ and a function $T:\mathcal Z_1\times\cdots\times\mathcal Z_k\to[0,1]$
such that $
    M=\sum_{z_1,\ldots,z_k} T(z_1,\ldots,z_k) \ M_1^{z_1}\otimes\cdots\otimes M_k^{z_k}.
$
Define the classical outcome space $\mathcal{Z} \coloneqq \mathcal {Z}_1\times\cdots\times\mathcal {Z}_k$ and,
for $z=(z_1,\ldots,z_k)\in\mathcal Z$, define the local outcome probabilities
$
    p_i(z_i)\coloneqq\Tr\!\big[M_i^{z_i}\,A_i(\rho_i)\big],
\
q_i(z_i)\coloneqq\Tr\!\left[M_i^{z_i}\,A_i(\sigma_i)\right].
$
Since the global output states are a product of states and the measurement is a product POVM at the level of outcomes,
the induced joint distributions factor as follows:
$
    P(z)\coloneqq\prod_{i=1}^k p_i(z_i),
\\
Q(z)\coloneqq\prod_{i=1}^k q_i(z_i).$
Moreover, the acceptance probability of the $\operatorname{LO}^*$ test $M$ is exactly a classical post-processing leading to the following equivalent formulation with expectation over $P$ and $Q$:
\begin{equation}\label{eq:accept-as-expectation}
\Tr\!\left[ M \left( \bigotimes_{i=1}^k A_i(\rho_i) \right) \right] = \mathbb{E}_{Z\sim P}[T(Z)],
\qquad
\Tr\!\left[ M \left( \bigotimes_{i=1}^k A_i(\sigma_i) \right) \right] = \mathbb{E}_{Z\sim Q}[T(Z)].
\end{equation}
For each $i$, define the privacy-loss random variable
$
    L_i(z_i)\coloneqq\log\frac{p_i(z_i)}{q_i(z_i)}.
$
By the QDP guarantee for  all measurement operators $M_i^{z_i}$ acting on $A_i$,
we have the pointwise bounds
\begin{equation}\label{eq:Li-bounded}
e^{-\varepsilon_i}\leq \frac{p_i(z_i)}{q_i(z_i)} \leq  e^{\varepsilon_i}
\quad \Longrightarrow \quad
|L_i(z_i)|\le \varepsilon_i\quad \forall z_i.
\end{equation}
Define the total privacy loss
$
    L(z)\coloneqq\sum_{i=1}^k L_i(z_i)=\log\frac{P(z)}{Q(z)}.
$
Let 
\begin{align}
    \mu_i & \coloneqq \mathbb{E}_{z_i\sim p_i}[L_i(z_i)] \\
   & = \sum_{z_i} p_i(z_i) \log\!\left( \frac{p_i(z_i)}{q_i(z_i)}\right)  \eqqcolon D_{\operatorname{KL}}(p_i\|q_i),
\end{align}
where the last equality uses the definition of KL-divergence (relative entropy for commuting states), and $  \mu\coloneqq\sum_{i=1}^k\mu_i.$

Since $\forall z_i, \ e^{-\varepsilon_i} \leq \frac{p_i(z_i)}{q_i(z_i)} \leq e^{\varepsilon} $, leading to the formulation of classical differential privacy, we have (\citep[Theorem~5]{DP_Sten}; see also \citep{BS16,harrison2025exact})
$
   \mu_i = D_{\operatorname{KL}}(p_i\|q_i) \leq \varepsilon_i \left( \frac{e^{\varepsilon_i}-1}{e^{\varepsilon_i} +1} \right).
$
Also, we get
\begin{equation}\label{eq:mu-bound}
\mu_i \leq \varepsilon_i \left( \frac{e^{\varepsilon_i}-1}{e^{\varepsilon_i} +1} \right)
\quad\Longrightarrow\quad
\mu \leq \sum_{i=1}^k \varepsilon_i \left( \frac{e^{\varepsilon_i}-1}{e^{\varepsilon_i} +1} \right).
\end{equation}
Now set
$
    t\coloneqq\sqrt{2\log\!\left(\frac1\delta\right)\sum_{i=1}^k\varepsilon_i^2},
\
\bar\varepsilon\coloneqq\mu+t,
$
and define the event
$
    G\coloneqq\{z\in\mathcal Z:  L(z)\leq \bar\varepsilon\}.
$

Under $P$, the coordinates $z_i$ are independent, hence $L_1(z_1),\ldots,L_k(z_k)$ are independent.
Also, by \eqref{eq:Li-bounded}, each centered variable
$X_i \coloneqq L_i(z_i)-\mathbb{E}_{P}[L_i(z_i)] \in [a_i, b_i]$ with $|b_i-a_i | \leq 2 \varepsilon_i$.
Then, Hoeffding's inequality yields the following with the choice of $t$:
\begin{align}
     P\!\left( L-\mathbb{E}_P[L] > t \right) &\leq  \exp\!\left(-\frac{2 t^2}{\sum_{i=1}^k (b_i-a_i)^2 }\right) \\
     & \leq \exp\!\left(-\frac{t^2}{2\sum_{i=1}^k \varepsilon_i^2}\right) = \delta, 
\end{align}
where we used $|b_i-a_i| \leq 2 \varepsilon_i$.

Since $\mathbb{E}_P[L]=\sum_i \mathbb{E}_{p_i}[L_i]=\mu$, this shows
\begin{equation}\label{eq:P-bad}
P(G^c)=P(L>\mu+t)\leq \delta.
\end{equation}
For any $T(\cdot)\in[0,1]$, 
with $P(z)=e^{L(z)}Q(z)$, we can write
\begin{equation}
    \mathbb{E}_{P}[T] = \sum_z P(z)T(z)=\sum_z Q(z)e^{L(z)}T(z)=\mathbb{E}_Q[e^{L}T].
\end{equation}
Now Split over $G$ and $G^c$ with $\mathbf 1_{B}$ is the indicator function on the event $B$:
\begin{align}
    \mathbb{E}_Q[e^L T]
& = \mathbb{E}_Q[e^L T\mathbf 1_G] + \mathbb{E}_Q[e^L T\mathbf 1_{G^c}] \leq e^{\bar\varepsilon}\mathbb{E}_Q[T] + \mathbb{E}_Q[e^L \mathbf 1_{G^c}],
\end{align}
where the last inequality holds since on $G$ we have $e^L\le e^{\bar\varepsilon}$, $0\le T\leq 1$, and $\mathbf 1_G \leq 1$.
Then, observe that
\begin{align}
    \mathbb{E}_Q[e^L \mathbf 1_{G^c}] &= \sum_z Q(z)e^{L(z)}\mathbf 1_{G^c}(z) = \sum_z P(z)\mathbf 1_{G^c}(z)  = P(G^c) \leq \delta,
\end{align}
where the second equality by $L(z)= \log\frac{P(z)}{Q(z)}$, and the last inequality
by \eqref{eq:P-bad}. Therefore,
$\mathbb{E}_P[T] \leq e^{\bar\varepsilon} \mathbb{E}_Q[T] + \delta,$
and by~\eqref{eq:accept-as-expectation} and~\eqref{eq:mu-bound}, we conclude the proof.
\end{proof}

Next, we derive an advanced composition result for the setting that holds even when $\delta_i \neq 0$ with the adversary performing measurements on ${\mathcal{M}}_{\operatorname{LO}*}$ by utilizing advanced composition for classical differentially private mechanisms in~\Cref{lem:Theorem_Classical_advanced_Comp}.

\begin{proposition}\label{prop:all_local_advanced}
    Let $A_i$ satisfy $(\varepsilon_i, \delta_i)$-QDP. Then for product neighbors $ \otimes_{i=1}^k \rho_i \sim \otimes_{i=1}^k \sigma_i$ and $M\in {\mathcal{M}}_{\operatorname{LO}^*}$, we have that 
    \begin{equation}
        \Tr\!\left[ M \left( \bigotimes_{i=1}^k A_i(\rho_i) \right) \right] \leq e^{\bar{\varepsilon}} \Tr\!\left[ M \left( \bigotimes_{i=1}^k A_i(\sigma_i) \right) \right] + \bar{\delta},
    \end{equation}
where for all $\delta \in (0,1)$ and
\begin{align}
    \bar{\varepsilon} & \coloneqq  \min \left\{ \sum_{i=1}^k \varepsilon_i , \  \sum_{i=1}^k \varepsilon_i \left( \frac{e^{\varepsilon_i}-1}{e^{\varepsilon_i} +1} \right) + \sqrt{2  \sum_{i=1}^k {\varepsilon_i^2} \ \min\left\{\log\!\left( \frac{1}{\delta} \right), \log\!\left(e+ \frac{ \sum_{i=1}^k \varepsilon_i^2 }{\delta} \right)  \right\}} \right \} \\ 
    \bar{\delta} & \coloneqq 1- (1-\delta) \prod_{i=1}^k (1-\delta_i).
\end{align}
\end{proposition}
\begin{proof}
Fix $M\in\mathcal M_{\operatorname{LO}^*}$. By definition, there exist POVMs
$\{M_i^{z_i}\}_{z_i\in\mathcal Z_i}$ and a function $T:\mathcal Z_1\times\cdots\times\mathcal Z_k\to[0,1]$
such that
$
    M=\sum_{z_1,\ldots,z_k} T(z_1,\ldots,z_k) \ M_1^{z_1}\otimes\cdots\otimes M_k^{z_k}.
$
Define the classical outcome space $\mathcal{Z} \coloneqq \mathcal {Z}_1\times\cdots\times\mathcal {Z}_k$ and,
for $z=(z_1,\ldots,z_k)\in\mathcal Z$, define the local outcome probabilities $p_i(z_i)\coloneqq\Tr\!\big[M_i^{z_i}\,A_i(\rho_i)\big],
\
q_i(z_i)\coloneqq\Tr\!\left[M_i^{z_i}\,A_i(\sigma_i)\right].$

Since the global output states are a product of states and the measurement is a product POVM at the level of outcomes,
the induced joint distributions factor as follows: $  P(z)\coloneqq\prod_{i=1}^k p_i(z_i),
\ 
Q(z)\coloneqq\prod_{i=1}^k q_i(z_i).$

Moreover, the acceptance probability of the $\operatorname{LO}^*$ test $M$ is exactly a classical post-processing leading to the following equivalent formulation with expectation over $P$ and $Q$:
\begin{equation}\label{eq:accept-as-expectation_2}
\Tr\!\left[ M \left( \bigotimes_{i=1}^k A_i(\rho_i) \right) \right] = \mathbb{E}_{Z\sim P}[T(Z)],
\qquad
\Tr\!\left[ M \left( \bigotimes_{i=1}^k A_i(\sigma_i) \right) \right] = \mathbb{E}_{Z\sim Q}[T(Z)].
\end{equation} \label{eq:close_to_classical}
 Recall that $A_i$ satisfies $(\varepsilon_i,\delta_i)$-QDP channel (with neighbors $\rho_i \sim \sigma_i$). Fix any POVM $\{M_i^{z_i}\}_{z_i\in\mathcal{Z}_i}$ and define the induced
classical mechanism $B_i$ that maps an input state $\omega$ to an outcome $z_i\in\mathcal{Z}_i$ with
$\operatorname{Pr}(B_i(\omega)=z_i):=\Tr[M_i^{z_i}A_i(\omega)]$. Then,  $B_i$ is $(\varepsilon_i,\delta_i)$-differentially private (classical) with respect to the same
neighboring relation on the input states.
To see that, let $S\subseteq\mathcal Z_i$ and define the measurement operator $M_S:=\sum_{z_i\in S}M_i^{z_i}$, which satisfies
$0\preceq M_S\preceq I$. Since $A_i$ is $(\varepsilon_i,\delta_i)$-QDP for neighboring inputs $\rho_i\sim\sigma_i$, we get 
$
    \Tr[M_S A_i(\rho_i)]\le e^{\varepsilon_i}\Tr[M_S A_i(\sigma_i)]+\delta_i.
$
By the definition of $B_i$,
$\Tr[M_S A_i(\rho_i)]=\Pr(B_i(\rho_i)\in S)$ and similarly for $\sigma_i$, so $B_i$ satisfies 
the classical DP as in~\Cref{def: DP}.

Now, considering the composition of classical DP mechanisms $(B_1, \ldots, B_k)$ obtained by measuring $A_i(\cdot)$ with the POVM $\{M_i^{z_i}\}_{z_i \in \mathcal{Z}_i}$, together with the classical outcomes from each mechanism ($z= (z_1,\ldots, z_k) \in \mathcal{Z}_1 \times \cdots \times \mathcal{Z}_k)= \mathcal{Z}$), and  by applying classical advanced-composition guarantees for this non-adaptive setting~\cite[Theorem~3.5]{KairouzOhVi17}, we obtain 
\begin{equation}\label{eq:use_of_classical}
    \forall S \subseteq \mathcal{Z}, \  P(S) \leq e^{\bar{\varepsilon}} \ Q(S) + \bar{\delta}.
\end{equation}
with $\bar{\varepsilon}, \bar{\delta}$ defined in the Proposition statement (see also~\Cref{lem:Theorem_Classical_advanced_Comp}) and 
$
    P(S) = \sum_{z \in S} P(z),  \\  Q(S) = \sum_{z \in S} Q(z).
$
With that, choose the following $S_t \subseteq \mathcal{Z}$ for a measurable function $T: \mathcal{Z} \to [0,1]$ (measurability holds since $\mathcal{Z}$ is a finite set since $\mathcal{Z}_i$ is finite for all $i\in\{1,\ldots, k\}$): For $t \in [0,1]$
\begin{equation}
    S_t \coloneqq \left\{ z: T(z) \geq t  \right\},
\end{equation}
which leads to $  P(S_t) \leq e^{\bar{\varepsilon}} \ Q(S_t) + \bar{\delta}.$
Also note that 
\begin{equation}
    \mathbb{E}_{Z\sim P}[T(Z)] = \int_{0}^1 P(S_t)\  \mathrm{d}t, \quad  \mathbb{E}_{Z\sim Q}[T(Z)] = \int_{0}^1 Q(S_t)\  \mathrm{d}t, 
\end{equation}
since 
\begin{align}
    \mathbb{E}_{Z\sim P}[T(Z)] & =\sum_{z \in \mathcal{Z}}  T(z) P(z)= \sum_{z \in \mathcal{Z}} \left(\int_{0}^1 \mathbf{1}_{\{z: T(z) \geq t\}} \mathrm{d}t \right) \  P(z) 
    &=\int_{0}^1  \sum_{z \in \mathcal{Z}} \mathbf{1}_{\{z: T(z) \geq t\}} P(z) \ \mathrm{d}t \\ 
    &=\int_{0}^1  P(S_t) \ \mathrm{d}t,
\end{align}
and similarly for $Q$ with $\mathbf{1}_{A}$ denoting the indicator function on set $A$.
With that, we arrive at the desired result $ \mathbb{E}_{Z\sim P}[T(Z)] \leq e^{\bar{\varepsilon}} \ \mathbb{E}_{Z\sim Q}[T(Z)] + \bar{\delta},$
 by using~\eqref{eq:accept-as-expectation_2} and concluding the proof since the above inequality holds for all $M \in \mathcal{M}_{\operatorname{LO}*}$ and product neighboring states.
\end{proof}
\begin{remark}[Improved Composition Results]
    One can also obtain strong advanced composition results by utilizing improved composition results in the classical setting (e.g.; \cite[Theorem~3.3]{KairouzOhVi17}) in~\Cref{prop:all_local_advanced} in the proof step~\eqref{eq:use_of_classical}.
    
\end{remark}

\section{Conclusion}
\label{sec:conclusion}

This work clarifies the landscape of composition guarantees for quantum differential privacy.
We showed that classical composition theorems fail in full generality for POVM-based approximate QDP, due to measurement incompatibility and correlated joint channels unique to the quantum setting.
At the same time, we demonstrated that these failures are not inherent to quantum channels per se, but arise from specific structural features absent in classical analysis. By restricting attention to tensor-product channels acting on product neighboring inputs, we recovered clean composition guarantees using a quantum moments accountant.
Our framework separates operator-level accounting from operational privacy guarantees, enabling advanced-composition-style bounds against arbitrary measurements. 

Several directions remain open.
It would be valuable to understand whether variants of the quantum moments accountant can handle broader classes of channels, such as factorized but non-tensor-product compositions, or whether approximate DP guarantees can be incorporated without losing additivity.
More generally, our results suggest that progress on quantum differential privacy will require careful alignment between mathematical structure and operational threat models, rather than direct transplantation of classical proofs. This work provides both technical tools and conceptual clarity for future investigations into privacy in quantum information processing.

\section*{Acknowledgments}

DA acknowledges helpful discussions with Guy Rothblum and Salil Vadhan regarding the composition of classical and quantum private mechanisms. DA also acknowledges support by the Simons Foundation (965342, D.A.) as part of the
Junior Fellowship from the Simons Society of Fellows.
TN acknowledges helpful discussions with Sujeet Bhalerao, Felix Leditzky, Vishal Singh, and Mark M.~Wilde on the composition of quantum private mechanisms.
 TN also acknowledges support from the
Department of Mathematics and the IQUIST Postdoctoral Fellowship from
the Illinois Quantum Information Science and Technology Center at
the University of Illinois Urbana-Champaign.


\bibliographystyle{alpha}
\bibliography{main}


\appendix

\section{Sandwiched R\'enyi DP to Approximate QDP}

\subsection{Sandwiched R\'enyi Divergence}
\label{subsec:sandwiched-renyi}

We adopt the standard \emph{sandwiched R\'enyi divergence} from quantum
information theory as our notion of R\'enyi divergence.

\begin{definition}[Sandwiched R\'enyi divergence \cite{muller2013quantum, wilde2014strong}]
\label{def:sandwiched-renyi}
Let $\alpha>1$.  The sandwiched R\'enyi divergence is
$
    \widetilde D_\alpha(\rho\|\sigma)
    :=
    \frac{1}{\alpha-1}
    \log
    \Tr\Big[
        \big(
            \sigma^{\frac{1-\alpha}{2\alpha}}
            \rho
            \sigma^{\frac{1-\alpha}{2\alpha}}
        \big)^\alpha
    \Big].
$
\end{definition}

When $\rho$ and $\sigma$ commute (e.g., in the classical case), the
sandwiched R\'enyi divergence reduces to the classical R\'enyi divergence:
if $\rho=\sum_i p_i\ket{i}\!\bra{i}$ and
$\sigma=\sum_i q_i\ket{i}\!\bra{i}$ with respect to a common eigenbasis,
then
$
    \widetilde D_\alpha(\rho\Vert\sigma)
    =
    \frac{1}{\alpha-1}
    \log \sum_i p_i^\alpha q_i^{1-\alpha}, $
the classical R\'enyi divergence.

\begin{definition}[Sandwiched R\'enyi DP (quantum)]
Let $\alpha>1$. A quantum channel $A$ satisfies $(\alpha,\varepsilon)$-\emph{sandwiched R\'enyi
differential privacy} if for all neighboring inputs $\rho\sim\sigma$, \\ $
\widetilde D_\alpha\!\big(A(\rho)\,\|\,A(\sigma)\big)\le \varepsilon.
$
\end{definition}

\begin{theorem}[From sandwiched R\'enyi DP to approximate QDP]
\label{thm:srdp-to-qdp}
If a quantum channel $A$ satisfies $(\alpha,\varepsilon)$-sandwiched R\'enyi DP for some
$\alpha>1$, then for all $\delta\in(0,1)$ it satisfies $(\varepsilon',\delta)$-QDP with
$
\varepsilon'
=
\varepsilon+\frac{\log(1/\delta)}{\alpha-1}.
$
\end{theorem}

\begin{proof}
Fix any POVM $M$. Since measurement is a CPTP map, by the data-processing inequality for the
sandwiched R\'enyi divergence (e.g.,~\citep{Beigi13}), \\
$
\widetilde D_\alpha\!\big(P_M(A(\rho))\,\|\,P_M(A(\sigma))\big)
\le
\widetilde D_\alpha\!\big(A(\rho)\,\|\,A(\sigma)\big)
\le
\varepsilon.
$
The resulting distributions are classical, so the standard conversion from R\'enyi DP to
$(\varepsilon',\delta)$-DP applies, yielding
$
\Pr[M(A(\rho))=1]
\le
e^{\varepsilon'}\Pr[M(A(\sigma))=1]+\delta,
$
with $\varepsilon'=\varepsilon+\log(1/\delta)/(\alpha-1)$.
\end{proof}

\section{Classical Differential Privacy}

\begin{definition}[Differential privacy (e.g., see~\citep{DworkMNS06,DworkKMMN06})] \label{def: DP}
Fix $\epsilon,\delta>0$. A~randomized mechanism\footnote{A randomized mechanism is described by a (regular) conditional probability distribution given the data, i.e., $P_{M|X}$.} ${B}: \mathcal{X} \to \mathcal{Z}$ is $(\epsilon,\delta)$-differentially private if for all $x \sim x'$ neighboring datasets and $S \subseteq \mathcal{Z} $ measurable, we have 
\begin{equation}
\Pr\big(B(x) \in S \big) \leq e^{\epsilon} \Pr \big(B(x') \in S \big) + \delta.\label{eq:dp_def}
\end{equation}
\end{definition}

\begin{lemma}\label{lem:Theorem_Classical_advanced_Comp}[Theorem~3.5 in~\citep{KairouzOhVi17}]
Let $B_i$ be a classical mechanism satisfying $(\varepsilon_i, \delta_i)$ classical differential privacy as in~\Cref{def: DP} with $\varepsilon_i >0$ and $\delta_i \in [0,1]$ for all $i \in \{1,\ldots, k\}$. Then the $k$-fold adaptive composition of these mechanisms (also for the non-adaptive setting) satisfy $(\bar{\varepsilon}, \bar{\delta})$-classical differential privacy with 
\begin{align}
    \bar{\varepsilon} & \coloneqq  \min \left\{ \sum_{i=1}^k \varepsilon_i , \  \sum_{i=1}^k \varepsilon_i \left( \frac{e^{\varepsilon_i}-1}{e^{\varepsilon_i} +1} \right) + \sqrt{2  \sum_{i=1}^k {\varepsilon_i^2} \ \min\left\{\log\!\left( \frac{1}{\delta} \right), \log\!\left(e+ \frac{ \sum_{i=1}^k \varepsilon_i^2 }{\delta} \right)  \right\}} \right \} \\ 
    \bar{\delta} & \coloneqq 1- (1-\delta) \prod_{i=1}^k (1-\delta_i),
\end{align}
for $\delta \in (0,1)$.
\end{lemma}

\end{document}